\documentclass[a4,12pt]{article}
\UseRawInputEncoding
\usepackage{amssymb}
\usepackage{amsfonts}
\usepackage{mathtools}
\usepackage{bm}
\usepackage[all]{xy}
\usepackage{graphicx}
\usepackage{verbatim}
\usepackage{wrapfig}
\usepackage{makeidx}
\usepackage{amscd}
\usepackage{url}
\usepackage{comment}
\usepackage{enumerate}
\usepackage{here}
\usepackage{latexsym}
\usepackage{array}
\usepackage{booktabs}
\usepackage{multirow}
\usepackage{mathrsfs}
\usepackage{geometry}
\usepackage{fancyhdr}

\renewcommand{\title}[1]{

\begin{center} \Large \bf #1 \end{center}
}

\renewcommand{\author}[2]{
 \begin{center} #1  \vspace{3mm} \\
  #2 \\
 \end{center}
\addvspace{\baselineskip}
}

\usepackage{amssymb}
\usepackage{amsmath}

\usepackage{amsthm}
\newtheorem{theorem}{Theorem}[section]
\newtheorem{proposition}[theorem]{Proposition}

\newtheorem{lemma}[theorem]{Lemma}
\newtheorem{example}[theorem]{Example}
\theoremstyle{definition}
\newtheorem{definition}[theorem]{Definition}

\theoremstyle{remark}
\newtheorem*{rem}{Remark}

\makeatletter
\@addtoreset{equation}{section}

\makeatother

\begin{document}

\baselineskip 5mm
\title{Quantization of Lie-Poisson algebra and
Lie algebra solutions of mass-deformed type IIB matrix model}
\author{${}^1$ Jumpei Gohara and~ ${}^1$ Akifumi Sako}{
${}^1$  Tokyo University of Science,\\ 1-3 Kagurazaka, Shinjuku-ku, Tokyo, 162-8601, Japan
}
\noindent
\vspace{1cm}

\abstract{
A quantization of Lie-Poisson algebras is studied.
Classical solutions of the mass-deformed Ishibashi-Kawai-Kitazawa-Tsuchiya (IKKT) matrix model can be constructed from semisimple Lie algebras whose dimension matches the number of matrices in the model.
We consider the geometry described by the classical solutions of the Lie algebras in the limit where the mass vanishes and the matrix size tends to infinity.
Lie-Poisson varieties are regarded as such geometric objects.
We provide a quantization called ``weak matrix regularization''
of Lie-Poisson algebras (linear Poisson algebras) on the algebraic varieties 
defined by their Casimir polynomials. 
This quantization is a generalization of matrix regularization, and neither faithfulness of the map nor the correspondence between integration and trace in the commutative limit is required.
Casimir polynomials correspond with Casimir operators of the Lie algebra by the quantization.
This quantization is a generalization of the method for constructing the fuzzy sphere.
In order to define the weak matrix regularization of the quotient space by the ideal generated by 
the Casimir polynomials, we take a fixed 
reduced Gr\"obner basis of the ideal. 
The Gr\"obner basis determines remainders of polynomials.
The operation of replacing these remainders with representation matrices of a Lie algebra roughly corresponds to a weak matrix regularization.
As concrete examples, we construct weak matrix regularization for 
$\mathfrak{su}(2)$ and $\mathfrak{su}(3)$. 
In the case of $\mathfrak{su}(3)$, we not only construct 
weak matrix regularization for the
quadratic Casimir polynomial, but also
construct weak matrix regularization for the cubic Casimir polynomial.
}

\section{Introduction}\label{sect1}

In M-theory and string theory, 
it has been proposed for a long time
that matrix models give their constructive formulation \cite{BFSS,IKKT3}.
Naturally, their physical quantities of classical solutions, etc. are given as matrices.
 (See for example \cite{IKKT2,fuzzya} and references therein.)
If we assume that the universe we live in can be obtained as a set of matrices, 
then at least the universe we observe can be approximated by a smooth manifold, 
so we need a correspondence between
 matrices and smooth spacetime.
In this paper, we will discuss quantization as a method 
for obtaining such a correspondence, 
and we will give a quantization of 
Poisson varieties with Lie-Poisson structures.\\

The matrix model given as a constructive formulation of type IIB superstring theory 
is called the type IIB matrix model or the IKKT (Ishibashi, Kawai, Kitazawa, and Tsuchiya) matrix model \cite{IKKT3}.
In this paper, we will refer to this model as the IKKT matrix model.
There are many studies that consider matrix solutions and eigenvalue distributions 
in the IKKT matrix model as spacetime.
One representative approach is the use of numerical methods using computers to simulate eigenvalue distributions, 
with many studies discussing, for example, the creation of a $3+1$-dimensional spacetime \cite{Kim:2011cr}.
Many results of numerical approaches have been summarized in \cite{Anagnostopoulos:2022dak}, 
and more can be found by consulting the references therein. 
There are also many studies that view the effective theory of the IKKT matrix model as a field theory on a noncommutative space.
The idea has its roots in earlier work, dating back to \cite{Gonzalez-Arroyo:1982hyq}.
The observation of the connection between IKKT and noncommutative geometry was made in \cite{IKKT2,Aoki:1999vr}.
The correspondence between finite-dimensional IKKT matrix models and gauge theories on noncommutative spaces was later revealed \cite{Ambjorn:1999ts,Ambjorn:2000nb,Ambjorn:2000cs}.
There are also several survey papers on this topic \cite{fuzzya,Ydri:2017ncg}.
As evidenced by the fact that the IKKT matrix model can be formulated 
as a zero-dimensional reduction of gauge theory, the model does not include a mass term, naively.
However, it is known that the only classical solution for the finite-dimensional IKKT matrix model 
without a mass term is given as a set of simultaneously diagonalizable matrices.
 (See the appendix of \cite{Steinacker:2017vqw}.)
The mass term in the IKKT matrix model was discussed 
in the new regularization of the model respecting the Lorentz symmetry in \cite{Asano:2024def}.
There is also research that suggests that the mass term itself exists effectively \cite{Laliberte:2024iof}.
There have also been reports of attempts to derive gravity using the IKKT matrix model, 
which is a model that undergoes mass deformations that preserve supersymmetry \cite{Komatsu:2024bop,Komatsu:2024ydh}. 
These deformations, however, are different from those that add a simple mass term.
This paper discusses classical solutions of the IKKT matrix model with a mass term 
and the space they are expected to represent.
\\

In those matrix models, a fuzzy space is an important concept that maps classical geometry to noncommutative matrix algebra. 
The fuzzy sphere is a typical example \cite{fuzzy1}. 
It provides a mapping from a set of polynomials on $S^2$ 
to the space of endomorphisms on an $N$-dimensional vector space,
which is a matrix algebra generated by an irreducible 
representation of $\mathfrak{su}(2)$. 
This mapping provides an approximate correspondence between the
Poisson brackets and commutators. 
Other fuzzy spaces as known fuzzy Riemann surfaces are constructed, due to the motivation of the membrane theory
in \cite{hoppe1,hoppe2,fuzzyb,arnlind2,arnlind3,Schneiderbauer}. 
In a similar study with Toeplitz operators, \cite{klimek} proved  that a Poisson algebra on any Riemann surface is produced in the limit $N\to \infty$, 
and \cite{berezin1} extended it to compact K\"{a}hler manifolds. 
These methods of mapping classical geometry to matrix algebra are called matrix regularization. 
Recently, fuzzy spaces have been studied using quasi-coherent states, 
and it has been shown that it is possible to extract the classical structure without the $N\to \infty$ limit \cite{Schneiderbauer,steinacker2}. 
This direction of research is developing in areas that are not limited to spaces with positive-valued measures.
Chany, Lu and Stern extended the fuzzy sphere studied in Euclidean to Minkowski and showed that it is a solution to the Lorentzian IKKT matrix model \cite{Stern1}.
Ho and Li showed that fuzzy $S^2$ and $S^4$ are candidates for the quantum geometry on the corresponding spheres in AdS/CFT correspondence \cite{ho}. 
\cite{Buric:2017yes} and subsequent papers by the same group report on fuzzy spaces in de Sitter and anti-de Sitter spaces.
\\

Fuzzy spaces have not only been studied in the context of the IKKT matrix model or the BFSS matrix model, but also as a direction in noncommutative geometry.
As already mentioned above, the history of fuzzy space begins with the proposal of the fuzzy sphere {\rm \cite{matrix1,fuzzy1}}. 
We can find more details in the {\rm \cite{matrix1,fuzzy1,fuzzyb,fuzzyc}}, as well as in the references within them.
Fuzzy torus is also constructed in a similar way.
(See for example \cite{Bal:2004ai}.) 
It is not possible to mention all examples, 
but it is possible to find summarized descriptions in, for example, Steinacker's textbook \cite{Steinacker_text}.
In this paper we construct a method to obtain 
fuzzy spaces of Lie-Poisson algebras.
As concrete examples, we also construct examples for 
$\mathfrak{su}(2)$ and $\mathfrak{su}(3)$. 
Closely related to the examples are, for example,
fuzzy $CP^2$  \cite{Nair:1998bp,Grosse:1999ci,ABIY,Balachandran:2001dd,Azuma_Bal_Nagao,Grosse:2004wm}.
The Lie-Poisson algebras treated in this paper are not limited to those corresponding to compact groups, but it is also possible to consider Lie algebras for noncompact groups.
The fuzzy spaces corresponding to 
noncompact groups are known, 
for example \cite{Hasebe:2012mz,Jurman:2013ota,Jurman:2017kkp,Sperling:2018xrm,Sperling:2019xar}.
These various studies constructed fuzzy spaces somewhat ad hoc for individual manifolds.
In contrast, the study by \cite{Rieffel:2021ykh} constructs general 
framework and mathematically precise statements by using 
$C^*$-metric, Gromov-Hausdorff-type distance.
The purpose of constructing a  fuzzy space of a coadjoint orbit 
of \cite{Rieffel:2021ykh} 
is a slightly similar to the purpose of this paper, which is to quantize
the space containing coadjoint orbit as a subspace.
The method used in this paper is simpler and easier to compute.
In this paper, we examine in detail in Section \ref{rev_sect6} the relationship between the quantization of coadjoint orbits 
and the quantization developed in this paper.
Various approaches to the correspondence between manifolds or algebraic varieties 
and the elements of matrix algebras or its subalgebras 
are expected to be explored in the future, and this paper contributes to that effort.
\\

Taking the classical (commutative) limit of a noncommutative manifold 
is often fraught with difficulty.
For example, as a well-known phenomenon, 
matrix regularization generates the same matrix algebra 
whether a two-dimensional torus is transformed into a fuzzy torus 
or a two-dimensional sphere into a fuzzy sphere.
(See \cite{deWitHoppeNicolai,arnlind,matrix1,fuzzy1, MadoreText}.)
Put another way, when reading a classical geometry (Poisson algebra)
from a matrix algebra, 
the Poisson algebra obtained depends on what classical (commutative) limit is taken
\cite{berezin1,bordemannA,Chu:2001xi}.
Sometimes, such a problem of determining the classical (commutative) limit 
is called an inverse problem.
Various approaches have been taken to the inverse problem 
of how to extract geometric properties from matrix algebras
\cite{shimada,berenstein,Schneiderbauer,ishiki1,ishiki2,asakawa}.
Therefore, a framework including any Poisson algebra and 
its quantizations is important in considering such issues to be investigated in a unified manner.
In \cite{Sako:2022pid},
a category of Poisson algebras and their quantized spaces is proposed, 
and a category-theoretic formulation of their classical limit is given.
As an example, a framework for obtaining a Poisson algebra as the classical limit of any semisimple Lie algebra was discussed in it.
\\

One purpose of this paper is to examine in detail the contents of Section 6 of that paper \cite{Sako:2022pid}.
There, the quantization of Lie-Poisson algebras and their inverse problems, as well as derivations based on the principle of least action, were discussed.
The definition of Lie-Poisson algebras in this paper is given as follows.
\begin{definition}
Let $x=(x_1, x_2, \cdots , x_d)$ be commutative variables.
Suppose that $(\mathbb{C}[x]  , \cdot , \{ ~ , ~ \})$ is a Poisson algebra. 
If the Poisson bracket acts as linearly, i.e., there exist structure constants $f_{ij}^k$ 
such that $\{ x_i  , ~ x_j \} = f_{ij}^k x_k$,
we call $(\mathbb{C}[x]  , \cdot , \{ ~ , ~ \})$ a Lie-Poisson algebra.
\end{definition}
The Lie-Poisson structure is introduced in \cite{S_Lie,Weinstein1983}.
The study of the quantization of Lie-Poisson algebras is given by Rieffel as a deformation quantization
\cite{rieffel_90}.
As a related study,
deformation quantization of polynomial Poisson algebras 
via universal enveloping algebra (generalizing that of Lie-Poisson structures) is discussed in
\cite{Penkava_Vanhaecke}.
In \cite{Sako:2022pid}, one of the authors discussed matrix regularization of 
Lie-Poisson algebras as an example of how the inverse problem of quantization 
can be treated in the category theoretical limit.
However, detailed discussions were omitted, 
and the case of taking a quotient by a nontrivial ideal was not sufficiently addressed.
In this paper, the details are given explicitly 
for a class of Lie algebras including all semisimple Lie algebras, 
and new concrete examples are constructed.
\\
\bigskip

Here, a synopsis of this paper is presented.\\

At first, the equivalence of the mass-deformed IKKT matrix model 
and the matrix models whose classical solutions are given as semisimple Lie algebras 
is shown in Section \ref{sect2}. (However, only the Bosonic part is mentioned.)
From the matrix model, it is shown that
the mass-deformed IKKT matrix model 
has a classical solution constructed from a
representation of any semisimple Lie algebra 
whose dimension matches the number of matrices in the matrix model.
This result was already derived by Arnlind and Hoppe in \cite{DiscreteMiniSurface}. 
These facts make it an important topic to investigate the spaces 
that emerge as commutative limits of the solutions. 
When taking the commutative limit, we simultaneously take the limit 
in which the dimension of the representation of the Lie algebra tends to infinity.
In this paper, we consider Lie-Poisson algebras as such spaces.
The algebra generated by a representation of a Lie algebra 
is constructed by quantization 
of a Lie-Poisson algebra. 
The space in which the Lie-Poisson algebra is defined is 
the variety described by $k$th-degree Casimir polynomials.
The space includes a coadjoint orbit of the Lie algebra.
The matrix regularization used in this paper is slightly different from the matrix regularization 
that is commonly used, and the conditions are weakened.
Therefore, to distinguish between the two, the term 
``weak matrix regularization'' is also used.
The difference is that weak matrix regularization is defined 
algebraically without using an operator norm, 
and only the condition corresponding to quantization is retained.
We do not require it to be approximately an algebra homomorphism, but in fact this condition is satisfied.
Details of the difference between matrix regularization and weak matrix regularization are discussed in Section \ref{3_2}.
An equivalence class of a polynomial defined on the variety corresponds to a matrix
by the weak matrix regularization.
\\

Here, we also mention an overview of how to construct the actual 
weak matrix regularization for 
a Lie-Poisson algebra ${\mathbb C}[x]/ I$ given in this paper, 
where $I$ is an ideal generated by Casimir polynomials.
Consider $[f(x)] \in {\mathbb C}[x]/ I$. 
A reduced Gr\"obner basis $G$ of $I$ is introduced. Then
for any $f(x) \in  {\mathbb C}[x]$ $f(x)= r (x) + h(x)$ is uniquely determined, 
where $h(x) \in I$ and $r(x) \notin I$. 
Roughly speaking, the weak matrix regularization is achieved 
by replacing all the variables $x_1, \cdots , x_d$ in this polynomial $r(x)$ 
with the corresponding elements of a matrix representation of the Lie algebra.
This process is a generalization of the method of constructing the fuzzy sphere.
In fact, as an example, we construct a fuzzy sphere by matrix regularization of $\mathfrak{su}(2)$
following the method proposed in this paper.
Furthermore, we consider an example of $\mathfrak{su}(3)$ 
and provide a matrix regularization of varieties 
determined not only by quadratic but also by cubic Casimir polynomials.
In weak matrix regularization, the existence of a nontrivial kernel of the map is generally allowed; that is, faithfulness of the map is not required. 
We describe this point in a slightly more concrete manner below.
Even in ordinary matrix regularization, it is common to have a nontrivial kernel. For example, higher-degree terms in polynomials are often contained in the kernel.
The matrix regularization considered in this paper is obtained as a sequence determined by the choice of representations of a Lie algebra; however, when a sequence of irreducible representations is chosen, an additional kernel depending on the representation is introduced, resulting in a larger kernel.
This implies that the classical space is not adequately reflected.
However, in Section \ref{rev_sect6} we discuss a method for constructing a more faithful quantization of the classical space 
by formulating the weak matrix regularization as a sequence of maps to reducible representations.
\\
\bigskip

The organization of this paper is as follows.
In Section \ref{sect2}, it is shown that 
any semisimple Lie algebra is a solution of the mass-deformed IKKT matrix model.
In Section \ref{sect3}, we summarize part of \cite{Sako:2022pid} 
and some mathematical facts, the minimum necessary to be used 
in constructing the weak matrix regularization of Lie-Poisson algebras in this paper.
Furthermore, the matrix regularization of varieties
corresponding to the case where the ideal is trivial, i.e., $\{ 0 \}$, is discussed.
In Section \ref{sect4}, the weak matrix regularization of varieties
that are determined by $k$th-degree Casimir polynomials is constructed.
In Section \ref{sect5}, as examples, we construct weak matrix regularizations 
for the cases where the Lie algebra is $\mathfrak{su}(2)$ and $\mathfrak{su}(3)$.
In Section {\ref{rev_sect6}, taking into account the well-known relationship between coadjoint orbits and representations, 
we extend the quantization obtained in Section \ref{sect4} with irreducible representations 
to the case of reducible representations.
Section \ref{sect6} provides a summary of this paper.
\\

This paper basically uses Einstein summation convention, 
but in cases where it is difficult to understand, 
the summation symbol $\sum$ will be used as appropriate.


\section{Lie algebras as solutions of IKKT matrix model}\label{sect2}
In this section, we discuss the relationship between mass-deformed IKKT matrix models 
and semisimple Lie algebras in order to clarify one of the motivations for this paper.
Let us consider the Bosonic part of the IKKT matrix model with a mass regularization term.
Using $N\times N$ Hermitian matrices $X^N_{\mu} ~(\mu = 1, \cdots , d, ~ N \in {\mathbb N} )$
and a mass $\hbar(N)$, we consider the action
\begin{align}
S_{IKKT}(\hbar^2(N)) [X]&= \mbox{\rm{tr}} \left(
-\frac{1}{4} [ X^N_{\mu} , X^N_{\nu}]^2 + \hbar^2(N) \frac{1}{2} X^{N \mu}X^N_{\mu}
\right) . \label{action_ikkt}
\end{align}
Here we take contraction 
as
$$[ X^N_{\mu} , X^N_{\nu}]^2 = \eta^{\mu \rho} \eta^{\nu \tau}
[ X^N_{\mu} , X^N_{\nu}][ X^N_{\rho} , X^N_{\tau}]
=\sum_{\mu,\nu,\rho,\tau}\eta^{\mu \rho} \eta^{\nu \tau}
[ X^N_{\mu} , X^N_{\nu}][ X^N_{\rho} , X^N_{\tau}]
 .$$
$\eta$ is a diagonal matrix and usually a Euclidean or a Minkowski metric.
In this paper, however, $\eta$ is a more general case 
that also includes $(n, d-n)$-type Minkowski metric;
$\eta= {\rm diag} (1,\cdots , 1 , -1, \cdots , -1)$.
The second term in (\ref{action_ikkt}) represents 
the deformation due to the mass term.
It should really be called the ``mass-deformed IKKT matrix model'', 
but in this paper we sometimes simply call it the IKKT matrix model.
(As already mentioned in Section \ref{sect1}, it is known that there are various classical solutions when the model is deformed with a mass term.
See for example \cite{Kim:2011ts,Kim:2012mw,Sperling:2019xar}.)
\\
\bigskip

\begin{rem}
A technical note should be made here.
Although we stated that $X^N_{\mu}$ is a Hermitian matrix
care must be taken when the matrix size $N$ is finite.
Because of the problem of the lack of finite-dimensional Hermitian representations 
for noncompact groups, in the following, 
$X_\mu$ corresponding to the negative sign of the metric is 
represented by an anti-Hermitian matrix $i X_\mu$ as in the Wick rotation.
In other words, since we consider a number of anti-Hermitian matrices corresponding 
to the negative eigenvalues, the model essentially corresponds to the Euclidean
IKKT matrix model in finite dimensions. 
To make this explicit, we refer to it as the ``Euclidean IKKT matrix model''.
By doing Wick rotation, the following argument can be applied to any Killing metric  
with negative eigenvalues. 
Therefore, for simplicity, we will describe the case of 
positive-valued metrics without the imaginary unit.\\
(However, this problem is not simple. There have been recent discussions on introducing gauge fixing without performing a Wick rotation, and there exist papers showing that the results differ between the Lorentzian and Euclidean cases \cite{Chou:2025moy}. Since delving into such delicate issues is not the purpose of this paper, we shall not discuss them further.)
\end{rem}
\bigskip

In \cite{Sako:2022pid}, classical limits of quantizations 
were discussed in the category 
``Quantum world'' which contains the whole of classical spaces (Poisson algebras) and their quantized spaces. 
In it, a matrix model which treats the Lie algebras 
as quantum spaces was introduced.
The matrix model was given as
\begin{align}
S_N(\hbar^2(N)) [X] &= \mbox{\rm{tr}} \left(
-\frac{1}{4} g^{\mu \rho} g^{\nu \tau} [ X^N_{\mu} , X^N_{\nu}][ X^N_{\rho} , X^N_{\tau}]
 + \hbar^2(N) \frac{1}{2} g^{\mu \nu} X^N_{\mu}X^N_{\nu}
\right) . \label{LieIKKT}
\end{align}
Here $g^{\mu \nu} \in {\mathbb R}$ represents the component of a real
symmetric nondegenerate matrix. In the following, $g$ is referred to as a metric for simplicity.
(There have been similar studies in the past that have the same equation of motion \cite{Ishii:2008tm,Kim:2003rza,Yang:2009pm}, but 
the actions themselves are different.) 
It is possible to show that the model (\ref{LieIKKT}) is actually equivalent to the IKKT matrix model (\ref{action_ikkt}).
We shall first look at the derivation of this fact.\\

When considering ordinary Riemannian  geometry, there exists a local coordinate transformation that can transform any nondegenerate symmetric positive-definite matrix 
into the Euclidean metric as
$g_{\mu \nu} dx^\mu dx^\nu =  \eta_{\mu \nu} dx^{\prime \mu} dx^{\prime \nu}$.
In the same way now, we diagonalize the metric $g$ and perform the variable transformation as follows.
\begin{align}
\sum_{\mu, \nu} {O^{-1}}_{\tau \mu} g^{\mu \nu} {O_{\nu \rho}} 
=  \lambda_\tau \eta^{\tau \rho} , \qquad  
\sum_{\nu} \sqrt{\lambda_{\tau}} {O^{-1 }}_{\tau \nu} X^N_\nu=
\sum_{\nu} \sqrt{\lambda_{\tau}} {O}_{\nu \tau} X^N_\nu
=: X^{N \prime}_{\tau} .
\label{var_changing}
\end{align}
Einstein summation convention is not used here.
In other words, no contraction is taken with respect to the subscript $\tau$.
Here $O$ is an orthogonal matrix to diagonalize $(g^{\mu \nu} )$, 
and $\lambda_\tau \eta^{\tau \tau}~(\tau = 1, \cdots , d)$ are eigenvalues of the matrix $(g^{\mu \nu} )$.
Note that $X^{N \prime}_{\tau}$ remain a Hermitian matrix.
Using this change of variables,
\begin{align}
g^{\mu \nu} X^N_{\mu}X^N_{\nu} =
\eta^{\mu \nu} X^{N \prime }_{\mu}X^{N \prime}_{\nu}
\label{mass_equiv}
\end{align}
and 
\begin{align}
g^{\mu \rho} g^{\nu \tau} [ X^N_{\mu} , X^N_{\nu}][ X^N_{\rho} , X^N_{\tau}]
= \eta^{\mu \rho} \eta^{\nu \tau}
[ X^{N \prime}_{\mu} , X^{N\prime}_{\nu}][ X^{N\prime}_{\rho} , X^{N\prime}_{\tau}] 
\end{align}
are easily obtained. Then we get 
\begin{align}
S_{IKKT}(\hbar^2(N)) [X^{\prime}] = S_N(\hbar^2(N)) [X] .
\end{align}
\begin{rem}
The path integral measure of the Bosonic part of the IKKT matrix model is 
given by
\begin{align*}
{\mathcal D} X :=
\prod_{\mu} \left(\prod_{i=1}^N d (X^N_{\mu})_{ii} \right)
\left( \prod_{i>j} d (X^{N Re}_\mu )_{ij} d (X^{N Im}_\mu )_{ij} \right).
\end{align*}
Here we calculate the Jacobian that appears in the above variable changing.
From (\ref{var_changing}),
\begin{align*}
d (X^N_{\mu})_{ij} = \sum_{\tau} \frac{1}{\sqrt{\lambda_{\tau}}} 
{O_{\mu \tau}} d (X^{N \prime}_{\tau})_{ij}.
\end{align*}
Note that
$$
\det \left( \frac{1}{\sqrt{\lambda_\nu}} {O_{\mu \nu}} \right)
=\prod_{\tau =1}^d \frac{1}{\sqrt{\lambda_\tau}}
\det \left( {O_{\mu \nu}} \right) = \sqrt{ \det (g_{\mu \nu} )},
$$
where $(g_{\mu \nu} )$ is the inverse matrix of $(g^{\mu \nu} )$.
Then we obtain 
\begin{align}
{\mathcal D} X = ( \det (g_{\mu \nu} ))^{\frac{N^2}{2}} {\mathcal D} X' .
\end{align}
So, if we define the Bosonic part partition functions 
$$Z_{IKKT} := \int {\mathcal D} X e^{-S_{IKKT}}
, \qquad 
Z_{Lie}:= \int {\mathcal D} X e^{-S_{N}},
$$
the relation between them is the following;
\begin{proposition}
The relation between the two Bosonic part partition functions given 
by the actions (\ref{action_ikkt}) and (\ref{LieIKKT}) is as follows.
\begin{align}
Z_{IKKT} = ( \det (g_{\mu \nu} ))^{\frac{N^2}{2}} Z_{Lie} .
\end{align}
\end{proposition}

\end{rem}
\bigskip

Next, we consider the case where the metric $g$ is given by the Killing metric of a Lie algebra.
The Killing metric is  expressed as
$g_{\mu \nu}=- f_{\mu \rho}^\tau f_{\nu \tau}^\rho$,
where $f_{\mu \rho}^\tau$ are structure constants 
 of some Lie algebra $\mathfrak{g}$.
Since the existence of this non-degenerate metric is equivalent to the semisimplicity of the Lie algebra, 
we require that the Lie algebra $\mathfrak{g}$ is semisimple in the following discussion of the Killing form as a metric.

As already shown in \cite{Sako:2022pid},
the equation of motion of (\ref{LieIKKT}) ,
\begin{align}
[X^{N \mu} ,  [ X^N_{\mu} , X^N_{\nu}]]= - \hbar^2(N) X^N_{\nu} , \label{EOM1}
\end{align}
has the solution such that
\begin{align}
 [ X^N_{\mu} , X^N_{\nu}] =  \hbar(N) f_{\mu \nu}^{\rho} X^N_\rho , \label{comm_rel}
\end{align}
with an orthogonal $\mathfrak{g}$ basis $\{ X^{N}_\mu \}$ satisfying $ {\rm tr} X^{N \mu} X_\nu^N = g^{\mu \tau}  {\rm tr} X^{N}_\tau X_\nu^N= c {\delta^\mu}_\nu $,
where $c$ is a constant.
In fact, we substitute the commutation relation (\ref{comm_rel}) for the left side of (\ref{EOM1}).
\begin{align*}
[X^{N \mu} ,  [ X^N_{\mu} , X^N_{\nu}]] &=
\hbar(N)f_{\mu \nu}^\rho g^{\mu \tau}[X^N_\tau , X^N_\rho ]
=\hbar^2(N) f_{\mu \nu}^\rho g^{\mu \tau} f_{\tau \rho}^\sigma X^N_\sigma
= \hbar^2(N) 
f_{\mu \nu}^\rho g^{\mu \tau} g^{\eta \sigma} f_{\tau \rho \eta} X^N_\sigma \\
&= \hbar^2(N) 
f_{\nu \mu }^\rho g^{\mu \tau} g^{\eta \sigma} f_{\eta \rho \tau } X^N_\sigma
= - \hbar^2(N){\delta_{\nu}}^\sigma X^N_\sigma = - \hbar^2(N)X^N_\nu .
\end{align*}
Here we used the property that $f_{\mu \nu \rho}:= f_{\mu \nu}^\tau g_{\tau \rho}$ is totally anti-symmetric
in the three indices when we chose the basis $ {\rm tr} X^{N \mu} X_\nu^N = c {\delta^\mu}_\nu $.
Therefore, we found that
such a set of generators $\{ X^N_{\mu} \}$ which form a representation of $\mathfrak{g}$ is a 
solution of the equation of motion of the action $S_N$.
\\

Summarizing the above considerations, the following theorem is obtained.
\begin{theorem}\label{thm_2_2}
Let $\mathfrak{g}$ be a $d$-dimensional semisimple Lie algebra.
The mass-deformed Euclidean IKKT matrix model 
has a classical solution constructed from a
representation of any $\mathfrak{g}$.
The solution $X^\prime$ is constructed by using a $\mathfrak{g}$ basis $X$;
\begin{align}
X^{N \prime}_{\tau}  = \sum_{\nu} \sqrt{\lambda_{\tau}} {O^{-1}}_{\tau \nu} X^N_\nu ,
\end{align}
where $O$ is an orthogonal matrix such that $\sum_{\mu, \nu} {O^{-1}}_{\tau \mu} g^{\mu \nu} {O_{\nu \rho}} 
=  \lambda_\tau \eta^{\tau \rho}$, and the  basis $X$ satisfies
$$ [ X^N_{\mu} , X^N_{\nu}] =  \hbar(N) f_{\mu \nu}^{\rho} X^N_\rho, 
\quad {\rm tr} X^{N \mu} X_\nu^N = c {\delta^\mu}_\nu ,
\quad
 g_{\mu \nu}=- f_{\mu \rho}^\tau f_{\nu \tau}^\rho .$$
\end{theorem}
The statement of this theorem is not new. 
Although it is phrased differently, essentially the same result was 
derived by Arnlind and Hoppe in \cite{DiscreteMiniSurface}.\\

As an aside, it should also be noted that structure constants 
are also changed as a result of change of variables.
To be more specific, if we define the structure constants $f_{\mu \nu}^{\prime \rho}$ as 
$ [ X^{N \prime}_{\mu} , X^{N \prime}_{\nu}] =  \hbar(N) f_{\mu \nu}^{\prime \rho} X^{N \prime}_\rho$, 
then this structure constants are given by
\begin{align}
f_{\mu \nu}^{\prime \rho} =\sqrt{\frac{\lambda_\mu \lambda_\nu}{\lambda_\rho}}
\sum_{\alpha, \beta, \sigma}O_{\alpha \mu} O_{\beta \nu}{O^{-1}}_{\rho \sigma} f^{\sigma}_{\alpha \beta} .
\end{align}
\bigskip

\begin{rem}
As noted above, we perform a Wick rotation and work in the Euclidean signature. 
Therefore, when constructing solutions of the mass-deformed IKKT matrix model with a metric
$\displaystyle \eta = {\rm diag} (\underbrace{1,\cdots , 1 }_{n}, \underbrace{-1, \cdots , -1}_{d-n}) $, 
the choice of Lie algebra must be based not only on having dimension 
$d$, but also on requiring that the signature of the eigenvalues 
of the Killing form matches that of $\eta$.
The Killing form has $n$ positive eigenvalues and 
$d-n$ negative eigenvalue  (i.e., the index of inertia is $(0,n, d-n)$).
\end{rem}
\bigskip

As described in this section, we have seen that the Lie algebra gives the solution 
of the IKKT matrix model with mass deformation.
Therefore, it is important to investigate classical objects 
for which quantization yields a Lie algebra or an algebra generated 
by a representation of that Lie algebra.
In the following, we consider this issue.


\section{ Quantization and preparations }\label{sect3}
In the previous section, we saw that any basis of an arbitrary semisimple Lie algebra 
corresponds to a classical solution of the mass-deformed IKKT matrix model.
The algebra generated by the representation of the Lie algebra corresponds to a noncommutative spacetime.
It is expected that Lie-Poisson algebras appear as the spaces 
 in commutative limits.
In Sections \ref{sect3} and beyond, we consider the mechanism
 that the noncommutative spacetime is obtained by quantization 
 of a Lie-Poisson algebra.
When we focus only on the classical solutions of the mass-deformed IKKT matrix model, the corresponding Lie algebra must be semisimple. 
However, the following discussion also applies to a broader class of 
Lie algebras. 
Even if the Lie algebra is not semisimple, the argument remains valid 
as long as the Casimir operator is proportional to the identity operator 
and its eigenvalues diverge in the limit 
where the dimension of the representation becomes infinite.
\\

In order to construct the quantization of Lie-Poisson algebras, 
we review some of the contents of \cite{Sako:2022pid,sako2024} below, 
and prepare for the discussion in Sections \ref{sect4} and beyond.

\subsection{Quantization}
The following definition of quantization is used in this paper.
\begin{definition}[Quantization map $Q$ \cite{Sako:2022pid} ]\label{Q}
Let ${A}$ be a Poisson algebra over a commutative ring $R$ over $ {\mathbb C}$, and
let $T_i$ be an $R$-module that is given by a subset of some 
$R$-algebra $(M , *_M)$. 
If an $R$-module homomorphism (linear map)  $t_{Ai} : {A} \to T_i  $ 
is equipped with a constant $\hbar(t_{Ai}) \in {\mathbb C}-\{0\}$ and 
satisfies 
\begin{align}\label{lie}
[t_{Ai}(f),t_{Ai}(g)]_M =\sqrt{-1}\hbar (t_{Ai})~ t_{Ai}(\{f,g\})
+\tilde{O}(\hbar^{1+\epsilon} (t_{Ai})) \quad 
(\epsilon >0 )
\end{align}
for arbitrary $f, g \in {A}$, where $[a ,b]_M := a *_M b - b *_M a $,
we call $t_{Ai}$ a quantization map or simply a quantization.
$\tilde{O}$ is defined in the Appendix \ref{ap1}.
We denote the set of all quantization maps by $Q$.
\end{definition}
In this paper, we use ${\mathbb C}$ or ${\mathbb C}[\hbar]$ as $R$ below.
One may wonder why the symbol $\tilde{O}$, which differs from 
the usual Landau symbol, is introduced in this definition of quantization. 
The reason why $\tilde{O}$ is used is that a norm is not introduced
in $T_i$, so the norm of the scalar ${\mathbb C}$ 
is used to determine the limit.
See also Appendix \ref{ap1} and Appendix \ref{AppenB} for a detailed explanation. 
\\

In \cite{Sako:2022pid}, the author defined the quantization as a part of the category of 
Poisson algebras and their quantizations, but in this paper, the discussion using categories 
is not necessary, so that part has been omitted.
\\

The above definition of quantization includes many kinds of quantizations.
For example, matrix regularization \cite{matrix1,arnlind},
fuzzy spaces \cite{fuzzy1}, and
Berezin-Toeplitz quantization 
\cite{berezin1,bordemannA,berezin2}
which have original ideas of the matrix regularization,
satisfy it.
In addition, the strict deformation quantization introduced by Rieffel 
\cite{rieffel1,rieffel2,strict2},
the prequantization \cite{prequantization0,prequantization1,prequantization2},
and Poisson enveloping algebras \cite{oh,oh2,um,can}
are also in $Q$.
(In \cite{Gohara:2019kkd,Jumpei:2020ngc},
we can see organized discussions about these quantization maps. 
The conditions for $Q$ are a part of the definition of 
pre-$\mathscr{Q}$ in \cite{Gohara:2019kkd,Jumpei:2020ngc}.)

\subsection{Matrix regularization and fuzzy spaces for Lie-Poisson algebras}\label{3_2}
The matrix regularization for Lie-Poisson algebras is introduced in \cite{Sako:2022pid}.
We review and refine it as ``weak matrix regularization''.
Furthermore, we discuss approximate algebra homomorphism 
between Lie-Poisson algebras and algebras generated by Lie algebras in this subsection.\\

Let $\mathfrak{g}$ be a finite dimensional Lie algebra.
Let $e=\{  e_1 , e_2 , \cdots ,e_d \}$ be a fixed basis of $\mathfrak{g}$
satisfying commutation relations $[ e_i , e_j ]= f_{ij}^k e_k $,
where $f_{ij}^k$ are structure constants of $\mathfrak{g}$. 
For this Lie algebra $\mathfrak{g}$ we introduce a sequence of irreducible representations 
$\rho^{\mu} : \mathfrak{g} \rightarrow gl(V^\mu ) ( \mu =1,2, \cdots )$ and
a sequence $\hbar(\mu) ( \mu =1,2, \cdots )$ with $\hbar({\mu} ) \neq 0$.
The case of reducible representations will be considered in Section \ref{rev_sect6}.
Here each $V^\mu$ is a finite dimensional vector space chosen as appropriate, and we put a condition
$\displaystyle \lim_{\mu \rightarrow \infty} \dim V^\mu = \infty$.
We denote the corresponding basis of $e$ by
\begin{align}\label{basis_matrix}
e^{(\mu)} =\{ 
\hbar(\mu) \rho^{\mu} (e_1 ), \hbar(\mu) \rho^{\mu} (e_2 ), 
\cdots , \hbar(\mu) \rho^{\mu} (e_d ) \}
= \{ 
e_1^{(\mu)} , e_2^{(\mu)} , \cdots  ,e_d^{(\mu)} \}.
\end{align}
Then they satisfy
\begin{align}
[ e_i^{(\mu)} , e_j^{(\mu)} ]= \hbar( \mu )f_{ij}^k e_k^{(\mu)} . \label{hbarCommRel}
\end{align}
The Lie algebra $\rho^{\mu}(\mathfrak{g})$ 
is constructed by this basis.\\

Next, we introduce a Poisson algebra corresponding to this Lie algebra.
There is a well-known way known as Kirillov-Kostant Poisson bracket,
that is the way constructing a Poisson algebra.
(See for example \cite{matrix1,Kostant,Weinstein_Lu,SemenovTianShansky,Alekseev:1993qs}.)
Let $x=(x_1, x_2, \cdots , x_d)$ be commutative variables.
We consider a Lie-Poisson algebra $(\mathbb{C}[x]  , \cdot , \{ ~ , ~ \})$ by
\begin{align}
\{ x_i , ~ x_j\}:= f_{ij}^k x_k , \label{x_brackets}
\end{align}
where $\displaystyle  f_{ij}^k \in {\mathbb C}$ are structure constants.
Concretely, this Poisson bracket is realized by
\begin{align}
\{ f , ~ g\}:= f \omega g := f \overleftarrow{\partial}_i \omega_{ij} 
\overrightarrow{\partial}_j g := ({\partial}_i f) \omega_{ij} 
({\partial}_j g) , \label{poisson_brackets}
\end{align}
where $\displaystyle \partial_i = \frac{\partial}{\partial x_i}$ and
$\omega_{ij} = f_{ij}^k x_k $.
It is easy to verify that (\ref{poisson_brackets}) satisfies the
Leibniz's rule and the Jacobi identity.
We denote the Poisson algebra $(\mathbb{C}[x]  , \cdot , \{ ~ , ~ \})$ by 
 $A_\mathfrak{g}$.
 We define degree of a monomial 
 $x^\alpha = (x_{1})^{\alpha_1} (x_{2})^{\alpha_2} \cdots  (x_{d})^{\alpha_d} $ by
 $\deg  x^\alpha :=|\alpha|:= \sum_{i=1}^d \alpha_i$. For the polynomial $f(x) = \sum_\alpha a_\alpha x^\alpha $, where
 $a_\alpha \in \mathbb{C}$, $\deg f(x) $ is defined by $\displaystyle \max_{a_\alpha \neq 0} \{  \deg x^\alpha \} $.
For multi-index, the notation $x^I =x_{i_1} x_{i_2} \cdots x_{i_m}$ 
is also often used below, where $\deg x^{I} = m=: |I|$.
 \\
 \bigskip

Next, let us construct quantization maps 
from the Lie-Poisson algebra $A_\mathfrak{g}$ to 
$T_{\mu}$ that is $\langle e^{(\mu)} \rangle$.
We denote the $R$-algebra generated by $ e^{(\mu)} $ by $\langle e^{(\mu)} \rangle$, here.
(If $\mathfrak{g}$ is a semisimple Lie algebra, then $T_\mu$ is given by $End (V^\mu )= gl (V^\mu )$.) 
In \cite{Sako:2022pid}, $T_{\mu}$ is regarded as a vector space that is $\langle e^{(\mu)} \rangle$
forgetting multiplication structure. 
However, we do not discuss the category $QW$ in this paper, 
so there is no problem treating $T_{\mu}$ as an algebra.
We choose a basis of $T_{\mu}:=\langle e^{(\mu)} \rangle$,
$E_1, E_2, \dots, E_D$, as polynomials of $e^{(\mu)}$.
Any polynomial of $e^{(\mu)}$ can be rewritten as $\hbar$ polynomial in $\langle \rho^\mu (e ) \rangle [\hbar ]$. 
So, a degree ${\rm deg}$ of any polynomial of $e^{(\mu)}$ can be defined by $\hbar$'s degree. 
Using $ E^i_{j_1 , \cdots , j_k}  \in {\mathbb C}$, $E_i~ (i=1, \cdots , D)$ are expressed as
$$ E_i  = \sum_k E^i_{j_1 , \cdots , j_k} e^{(\mu)}_{j_1} \cdots  e^{(\mu)}_{j_k}
= \sum_k \hbar^k(\mu)  E^i_{j_1 , \cdots , j_k} \rho^\mu (e_{j_1}) \cdots  \rho^\mu (e_{j_k})
, $$
where $E^i_{j_1 , \cdots , j_k}$ is independent of $\hbar$.
Note that Einstein summation convention is used for each $j_l$.
For each $E_i$, such expression given by $e^{(\mu)}$ is not unique in general.
We chose an expression that minimizes $\displaystyle \max_{E^i_{j_1 , \cdots , j_k}\!\!\!\!\!\!\!\!\!\!  \neq 0 } \{k \}$.
Then there exists the degree of $\hbar$ of $E_i$ i.e., 
$\displaystyle \deg E_i := \max_{E^i_{j_1 , \cdots , j_k}\!\!\!\!\!\!\!\!\!\!  \neq 0 } \{k \}$.
The highest degree of $\{E_1, E_2, \dots, E_D \}$ is denoted by $n_\mu$, i.e.,
$n_\mu  = \max \{ {\deg}E_1, \cdots , {\rm deg}E_D \}$.
$n_\mu$ does not depend on the choice of $\{E_1, E_2, \dots, E_D \}$,
because another $E'_1, E'_2, \dots, E'_D $ can also be described as linear combinations of the original $E_1, E_2, \dots, E_D $.\\

In considering matrix regularization, the quantization of the target space by a finite matrix algebra (and its subalgebras), the following properties are characteristic.
\begin{lemma}\label{lemma3_2}
Let the vector $E_{i_1} E_{i_2} \cdots E_{i_k} $ be represented by a linear 
combination of $\{E_1, E_2, \dots, E_D \}$ with each component $c_i$.
\begin{align*}
E_{i_1} E_{i_2} \cdots E_{i_k}  = \sum_{j} c_{j}  E_j .
\end{align*}
If $\sum_{l=1}^k \deg E_{i_l} > n_\mu $, then $c_i $ is a polynomial consisting of terms of degree $1$ or higher in $\hbar(\mu) $.
As a similar claim, for $k > n_{\mu}$ and
\begin{align*}
e^{(\mu)}_{i_1} \cdots e^{(\mu)}_{i_k} = \sum_{j} c_{j}  E_j ,
\end{align*}
then $c_i $
 is a polynomial consisting of terms of degree $k- n_\mu$ or higher in $\hbar(\mu) $.
\end{lemma}
\begin{proof}
It follows from the fact that the degree of $\hbar$ of $E_{i_1} E_{i_2} \cdots E_{i_k} $ is greater than $n_\mu$, but the maximum degree of $E_i$ is $n_\mu$.
\end{proof}

\begin{definition}\label{def_FuzzyQ} 
We define a linear function
$q_\mu : A_\mathfrak{g} \rightarrow T_\mu $
by
\begin{align} \label{q_mu}
\sum_I f_I x^I := \sum_k f_{i_1, \cdots , i_k} x_{i_1} \cdots x_{ i_k} \mapsto 
\sum_{|I|\le n_\mu} f_I e^{(\mu)}_{(I)} = \sum_k^{n_\mu} f_{i_1, \cdots , i_k} e^{(\mu)}_{(i_1, \cdots , i_k)} ,
\end{align}
where $f_{i_1, \cdots , i_k} \in \mathbb{C}$ is completely symmetric,
\begin{align*}
e^{(\mu)}_{(I)}:= e^{(\mu)}_{(i_1, \cdots ,i_k)}:=
\frac{1}{k!}\sum_{\sigma \in Sym(k)}
e^{(\mu)}_{i_{\sigma(1)}} \cdots e^{(\mu)}_{ i_{\sigma(k)}} ,
\end{align*}
and we require that the multiplicative identity of $A_\mathfrak{g}$ 
maps to the unit matrix in $T_\mu$ i.e.,  $q_\mu (1) = Id \in gl(V^\mu)$.
\end{definition}

\begin{rem}
The above $n_\mu$ is introduced to give a boundary that determines the kernel of $q_\mu$.
Although it is defined by $n_\mu  = \max \{ {\rm deg}E_1, \cdots , {\rm deg}E_D \}$ here, 
there is no particular reason why it has to be this way.
The important points are that $n_\mu$ 
is defined as a sequence that increases with the dimension of $V^\mu$,  
and that Lemma \ref{lemma3_2} is satisfied.
But the details are not essential.
For example, it is also possible to define $n_\mu$ by using multi-degree as follows.
We denote multi-degree of $E_i$ by ${\rm multdeg} E_i$. (In Appendix \ref{appendix_grobner}, we can see a definition of multi-degree.)
A $n_\mu$ can be defined as the highest multi-degree of $\{E_1, E_2, \dots, E_D \}$ i.e.
$n_\mu = \max \{ {\rm multdeg}E_1, \cdots , {\rm multdeg}E_D \}$.
It is possible to reduce the number of more degenerate elements of $q_\mu$.
However, this definition has not been chosen for simplicity.
Therefore,  the notation $\sum_k^{n_\mu}$ used here and in the following should be 
more accurately interpreted as 
$\sum_{(i_1, \cdots , i_k ) \in Dom}$, where 
\begin{align}
Dom= \{ (i_1, \cdots , i_k)~ | ~x_{i_1} \cdots x_{i_k}  \in A_\mathfrak{g}  \backslash \ker q_\mu \}. \label{Dom}
\end{align} 
\end{rem}
\bigskip

By definition, this quantization $q_\mu$ satisfies the following.
\begin{proposition}\label{sect5_thm_2}
Let $q_\mu : A_\mathfrak{g} \rightarrow  T_\mu = \langle e^{(\mu)}\rangle$ be a linear function 
of Definition \ref{def_FuzzyQ}. Then it satisfies
\begin{align*}
[q_{\mu} ( f ) , q_{\mu} ( g )]
= \hbar ({\mu}) q_{\mu} ( \{ f , g \} ) + \tilde{O}(\hbar^2({\mu}))
\end{align*}
for $\forall f,g \in A_\mathfrak{g} $. In other words, $q_\mu \in Q$.
\end{proposition}
The proof is given in \cite{Sako:2022pid}, 
however, there is no need to refer to it,
as we will give a more detailed discussion soon 
in a slightly different framework of ``weak matrix regularization''.
Since $q_\mu$ maps $n$-degree polynomials to 
$n$-degree quantities in $T_\mu$, 
this proposition looks almost a trivial assertion.
In other words, there is the inability to distinguish between $\hbar$, 
which represents noncommutativity, and $\hbar$, 
which originates from the degree of the polynomial.
This makes it insufficient to use the degree of 
$\hbar$ as a meaningful measure of noncommutativity. 
In order to address this issue, we introduce a ``weak matrix regularization'' 
in which the noncommutativity is explicitly manifested in terms of  $\hbar$. 
Furthermore, the condition for noncommutativity expressed via 
the asymptotic homomorphism that appears at the end of this section 
makes this point even more explicit. \\
The quantization from a Lie-Poisson algebra to a matrix algebra or a subset of a matrix algebra, such as $q_\mu \in Q$, 
which is included in the quantization $Q$ is regarded as a matrix regularization in this paper. 
Since its construction is a certain generalization of the matrix regularization of Madore\cite{fuzzy1} or 
de Wit-Hoppe-Nicolai \cite{deWitHoppeNicolai}, this quantization 
is regarded as a matrix regularization.
The target of matrix regularization is called a Fuzzy space.
There is not necessarily a consensus on a single definition of ``matrix regularization''. 
Commonly used definition of matrix regularization is given in Appendix \ref{AppenB}.
The definition used in this paper is less restrictive than the one in Appendix \ref{AppenB}.
We define weak matrix regularization here.

\begin{definition}[Weak matrix regularization]
Consider a Poisson algebra $A$ and a
sequence of subalgebras $B^m \subset End (V^m) ~(m=1,2, \cdots )$ where
$V^m$ is a finite dimensional vector space and $\lim_{m\rightarrow \infty} \dim V^m = \infty$.
Let $x=(x_i, \cdots , x_d)$ be a generator of $A$.
Let $q_m : A \to B^m ~ (m=1,2, \cdots )$ be a sequence of quantizations such that 
$\hbar(m):=\hbar(q_m)$ tends to zero as the dimension of $V^m$ tends to infinity.
For each $q_m$ and $\forall f,g \in A $,
when there exists some 
$$P_m = \sum_l 
a^m_{i_1, \cdots , i_l} q_m(x_{i_1}) \cdots  q_m(x_{i_l}) \in B^m,$$ 
where $a^m_{i_1, \cdots , i_l} \in {\mathbb C}$ and
$P_m$ is $\tilde{O}(\hbar^k(m)) ~(k\ge 0)$
satisfying
\begin{align}
[q_{m} ( f ) , q_{m} ( g )]
= \hbar (m) q_{m} ( \{ f , g \} ) + \hbar^2({m}) P_m,
\label{3_mat_reg_proof}
\end{align}
we call $ \{ q_m \}$ a weak matrix regularization. 
(Einstein summation convention is used, and the summation symbol 
$\sum_{i_1 , \dots , i_l}$ is omitted.)
We also simply say that $q_m$ is a weak matrix regularization.
\end{definition}


\begin{theorem}
Consider a Poisson algebra $A_{\mathfrak{g}}$ and a
sequence of subalgebras 
$T_\mu \subset End (V^\mu) ~(\mu=1,2, \cdots )$ defined above.
Let $q_\mu : A_{\mathfrak{g}} \to T_\mu ~ (\mu =1,2, \cdots )$ be a sequence
of quantizations
defined by Definition \ref{def_FuzzyQ}.
Suppose that $\{\hbar(\mu)\}$ is a sequence such that 
$\hbar(\mu) \rightarrow 0$ as $\dim V^\mu \rightarrow \infty$ .
Then $q_\mu$
is a weak matrix regularization.
\end{theorem}
\begin{proof}
If we can show that for $\forall f,g \in A_\mathfrak{g} $
there exist $P = \sum_i^D c_i(\hbar(\mu)) E_i  \in T_\mu$
with $c_i(\hbar) \in {\mathbb C}[\hbar (\mu )]$ such that
\begin{align}
[q_{\mu} ( f ) , q_{\mu} ( g )]
= \hbar ({\mu}) q_{\mu} ( \{ f , g \} ) + \hbar^2({\mu}) P,
\label{3_mat_reg_proof}
\end{align}
then $q_\mu$ is a weak matrix regularization.
By Definition \ref{def_FuzzyQ}, 
$\displaystyle 
q_\mu (x^{I_k})= e^{(\mu)}_{(I_k)} 
= \frac{1}{k!} \!\!\!\! \sum_{\sigma \in Sym(k)} \!\!\!\!
e^{(\mu)}_{i_{\sigma (1)}} \cdots  e^{(\mu)}_{i_{\sigma (k)}}
$ for $|I_k| \le n_\mu$, and $\displaystyle 
q_\mu (x^{I_k})=0$ for $|I_k| > n_\mu$.

When $|I_k| +|J_m|  \le n_\mu +1 $, 
degree of $\{ x^{(I_k)} , x^{(J_m)} \}$ is $|I_k| +|J_m|-1$.
So, $[ e^{(\mu)}_{(I_k)}  , e^{(\mu)}_{(J_m)} ] = \hbar (\mu)
q_{\mu} ( \{ x^{(I_k)} , x^{(J_m)}  \} ) + \hbar^2(\mu) P_1$.
Here we use $P_k = \sum_i^D c^k_i(\hbar(\mu)) E_i  \in T_\mu$, that is,
each coefficient of base $E_i$ is a polynomial in  $\hbar(\mu)$.
When $|I_k| +|J_m|  > n_\mu +1 $,
$q_\mu ( \{ x^{(I_k)} , x^{(J_m)} \} )=0$.
$[ e^{(\mu)}_{(I_k)}  , e^{(\mu)}_{(J_m)} ] = \hbar^2(\mu) P_2$,
since the commutator makes $\hbar(\mu)$ and
the other $\hbar(\mu)$  arise from Lemma \ref{lemma3_2}.
Then, for $\displaystyle f = \sum_{I_k} f_{I_k} x^{I_k} $
and $\displaystyle g = \sum_{J_m} g_{J_m} x^{J_m}$, 
\begin{align}
[q_{\mu} ( f ) , q_{\mu} ( g )] &=
\sum_{|I_k|, |J_m| \le n_\mu} f_{I_k} g_{J_m} [ e^{(\mu)}_{(I_k)}  , e^{(\mu)}_{(J_m)} ] 
\notag \\
&=
\sum_{|I_k|+ |J_m| \le n_\mu+1} f_{I_k} g_{J_m} [ e^{(\mu)}_{(I_k)}  , e^{(\mu)}_{(J_m)} ] 
+ \hbar^2 (\mu) P_3
\notag \\
&=  \hbar (\mu) q_{\mu} ( \{ f , g \} ) + \hbar^2 (\mu) P_4. 
\end{align}
The $q_{\mu} ( \{ f , g \} ) $ is symmetrized, but reordered in the same order as 
$[ e^{(\mu)}_{(I_k)}  , e^{(\mu)}_{(J_m)} ] $ in the second line 
by using the commutation relation 
of the Lie algebra, and all the extra terms from the reordering 
are included in the second term.
\end{proof}

The following proposition is evident from the fact that 
$\hbar$ appears in the commutator product.
\begin{proposition}\label{prop3_7}
For any $f,g \in A_{\mathfrak{g}}$ with $\deg f + \deg g \le n_\mu$, there exists $P = \sum_i^D c_i(\hbar(\mu)) E_i  \in T_\mu$
with $c_i(\hbar) \in {\mathbb C}[\hbar (\mu )]$ satisfying
$$ q_{\mu} ( f )  q_{\mu} ( g ) = q_{\mu} ( f  g ) +\hbar(\mu) P.  $$
\end{proposition}
\bigskip


In the definition of weak matrix regularization, 
the restrictions on correspondence with the classical (commutative) limit are relaxed.
How to choose the limit that brings $\hbar$ close to zero, or so-called classical (commutative) limit, 
is a delicate matter. 
Indeed, these classical limits lead to various classical Poisson algebras \cite{Chu:2001xi,Sako:2022pid}. 
Therefore, some restrictions on the classical limit are not  discussed in this paper.
For example, quantization maps are not equipped with algebra homomorphism, 
but it is required to be homomorphic in the limit where $\hbar$ approaches zero in Appendix \ref{AppenB}. 
This paper does not focus on a rigorous proof of 
algebra homomorphism in the limit by using operator norm, 
but Proposition \ref{prop3_7} shows that $q_\mu$ has a similar property.
Furthermore, we shall discuss this issue briefly by reversing the direction of the mapping,
in the remainder of this section.
We need a little preparation.
First, let us introduce an enveloping algebra with variable $\hbar$.
\\

%

From a set  $X= \{  X_1  , \cdots , X_d \}$ we make free monoid $\langle X \rangle = \{1\} \cup \{X_{i_1} X_{i_2} \cdots X_{i_n} ~| ~ n \in \mathbb{N} ,  X_{i_j} \in X \}$
and  free ${\mathbb C}$-algebra 
$${\mathbb C} \langle X \rangle = \left\{ \sum_n \sum_{ {i_1}, {i_2}, \cdots ,{i_n} } a_{{i_1}, {i_2}, \cdots ,{i_n}}  X_{i_1} X_{i_2} \cdots X_{i_n} \right\}. $$

The usual enveloping algebra of $\mathfrak{g}$, $ \mathcal{U}_\mathfrak{g} $,
is an algebra of all polynomials of $X_1  , \cdots , X_d $ with relations
$ X_i  X_j  - X_j  X_i  \sim [X_i ~ , ~ X_j ]:= f_{ij}^k X_k $.
This is a canonical definition of enveloping algebra.
We introduce a slightly changed algebra.
Let $ \mathcal{U}_\mathfrak{g} [ \hbar ]$ be an
algebra $R \langle X \rangle  $ over $R={\mathbb C}[\hbar]$
divided by a two-sided ideal $I$ generated by
\begin{align}
&X_i   X_j  - X_j   X_i  -  [X_i ~ , ~ X_j ] := X_i   X_j  - X_j   X_i  - \hbar f_{ij}^k X_k .
\label{commute_U_h} \\
& I := \left\{ \sum_{i,j,k,l}  a_{k} ( X_i   X_j  - X_j   X_i  -  [X_i ~ , ~ X_j ] ) b_{l} ~ | ~ a_{k} , b_{l} \in R \langle X \rangle  \right\}.
\end{align}
 Note that $\hbar$ is introduced in the relation. Let us introduce the following.
\begin{align*}
\mathcal{U}_\mathfrak{g}[\hbar ]
:= R \langle X \rangle  / I .
%
\end{align*}
This enveloping algebra $\mathcal{U}_\mathfrak{g}[\hbar ]$ allows for
discussion of the rest of this section and the next section.\\
\bigskip

Every $E_i$, which is the basis of $T_\mu$, has an expression of 
\begin{align}\label{SymEi}
E_i = \sum_k^{n_\mu} 
E^{(i)}_{j_1, \cdots , j_k} e^{(\mu)}_{(j_1, \cdots , j_k)} ,
\end{align}
where $E^{(i)}_{i_1, \cdots , i_k} \in {\mathbb C}$ is completely symmetric coefficient.
(In this paper, Einstein summation convention is used, 
so the above expression means 
$E_i = \sum_k^{n_\mu} 
\sum_{ j_1, \cdots , j_k } 
E^{(i)}_{j_1, \cdots , j_k} e^{(\mu)}_{(j_1, \cdots , j_k)}.$)
This follows from the following Lemma,
which is essentially equivalent to a part of Poincar\'e - Birkhoff - Witt (PBW) theorem.
\begin{lemma}\label{lemma3_4}
Let $ \mathcal{U}_\mathfrak{g}[\hbar ] $ be the enveloping algebra of Lie algebra $\mathfrak{g}$
defined above. 
Then $\forall X_{i_1} X_{i_2} \cdots X_{i_k} \in \mathcal{U}_\mathfrak{g}[\hbar ] $ can uniquely 
be written as  
\begin{align}
X_{i_1} X_{i_2} \cdots X_{i_k} = X_{(i_1, \cdots ,i_k)} + \sum_{l=0}^{k-1} 
 a_{j_1, \cdots , j_l} X_{(j_1, \cdots ,j_l)},
\end{align}
where $a_{j_1, \cdots , j_l} \in \mathbb{C}$ is completely symmetric, and 
\begin{align*}
X_{(i_1, \cdots ,i_l)}:=
\frac{1}{l!}\sum_{\sigma \in Sym(l)}
X_{i_{\sigma(1)}} \cdots X_{ i_{\sigma(l)}} .
\end{align*}
\end{lemma}
\begin{proof}
We shall show by mathematical induction.
$k=1$ is trivial.
Assume that it is valid up to the $l$-th.
It is sufficient to show that it is true for $X_{(i_1, \cdots ,i_l)}X_{i_{l+1}}$.
\begin{align}
X_{(i_1, \cdots ,i_l)}X_{i_{l+1}}
=&\frac{1}{l!}\sum_{\sigma \in Sym(l)}
X_{i_{\sigma(1)}} \cdots X_{ i_{\sigma(l)}} X_{i_{l+1}}
\notag \\
=&\frac{1}{l+1}\frac{1}{l!}\sum_{\sigma \in Sym(l)}
\sum_{k=1}^{l+1}
\left\{
(X_{i_{\sigma(1)}} \cdots X_{ i_{\sigma(k-1)}}) X_{i_{l+1}}
(X_{i_{\sigma(k+1)}} \cdots X_{ i_{\sigma(l)}} ) X_{ i_{\sigma(k)}} 
\right.
\notag \\
& - (X_{i_{\sigma(1)}} \cdots X_{ i_{\sigma(k-1)}}) 
[ X_{i_{l+1}} ~, ~ (X_{i_{\sigma(k+1)}} \cdots X_{ i_{\sigma(l)}} ) ] X_{ i_{\sigma(k)}}  \notag \\
& -\left.
 (X_{i_{\sigma(1)}} \cdots X_{ i_{\sigma(k-1)}}) 
 [ (X_{i_{\sigma(k+1)}} \cdots X_{ i_{\sigma(l)}} )X_{i_{l+1}}  ~, ~X_{ i_{\sigma(k)}} ] ~ \right\} . \label{3_10}
\end{align}
The terms that contain $[~ ,~ ]$ can be written as expressions of degree $l$ or lower, and are uniquely symmetrized by the assumption. Therefore, it is sufficient to look at the first term on the right-hand side of (\ref{3_10}).
The first term is written 
by
\begin{align*}
\frac{1}{(l+1)!}\sum_{\sigma \in Sym(l+1)}
X_{i_{\sigma(1)}} \cdots X_{ i_{\sigma(l+1)}}
= X_{(i_1, \cdots ,i_{l+1})}.
\end{align*}
The lemma was thus proved.
\end{proof}
To replace $X_i$ by $e^{(\mu )}_i$ 
each base $E_i $ is expressed as (\ref{SymEi}).
Because $E_i $'s  expression that represents as (\ref{SymEi}) 
is not uniquely determined,
we fix one expression of each $E_i$ by some (\ref{SymEi}).
Then any $f = \sum_i a^i E_i \in T_\mu = \langle e^{(\mu)}\rangle$
is uniquely written by
\begin{align}
f 
=  \sum_k^{n_\mu}  \sum_i a^i 
E^{(i)}_{j_1, \cdots , j_k} e^{(\mu)}_{(j_1, \cdots , j_k)}
= \sum_k^{n_\mu} 
f^{(k)}_{(j_1, \cdots , j_k)} e^{(\mu)}_{(j_1, \cdots , j_k)},
\end{align}
where 
$\displaystyle f^{(k)}_{(j_1, \cdots , j_k)}= \sum_i a^i 
E^{(i)}_{j_1, \cdots , j_k}  \in R$.
Using this fixed expression of $\{ E_i \}$,
let us introduce a linear map
$\phi_\mu : T_\mu \rightarrow  A_\mathfrak{g}$ such that
\begin{align} \label{phi_mu}
\phi_\mu ( f )=  
\sum_k^{n_\mu} 
f^{(k)}_{(j_1, \cdots , j_k)} x_{j_1} \cdots x_{ j_k} 
\end{align}
for $\displaystyle f= \sum_k^{n_\mu} 
f^{(k)}_{(j_1, \cdots , j_k)} e^{(\mu)}_{(j_1, \cdots , j_k)}$. In short, for $e^{(\mu)}_{(i_1, \cdots , i_k)}$
in (\ref{SymEi})
\begin{align*}
\phi_\mu (  e^{(\mu)}_{(i_1, \cdots , i_k)}   )=  x_{i_1} \cdots x_{ i_k} .
\end{align*}
It is clear from this definition, that
\begin{align}
q_\mu \circ \phi_\mu  = Id  
\end{align}
is satisfied.
Using this map $\phi_\mu$, we compare how similar
the algebra $T_\mu$ and the Poisson algebra 
${\mathbb C}[x]$ as algebras. 
In this paper, ``asymptotic algebra homomorphic map'' is defined as follows.
\begin{definition}[Asymptotic algebra homomorphism]
\label{def_weak_MR_asym_hom}
Let $A$ be a fixed Poisson algebra. 
Consider a sequence of vector spaces $V^m~ (m=1,2, \cdots )$ 
and a sequence of subalgebras of $End (V^m)$ denoted by $B^m ~ (m=1,2, \cdots )$.
Consider a weak matrix regularization $\{ q_m : A \to B^m \subset End (V^m)\}$.
Each noncommutativity parameter $\hbar(q_m)$ is abbreviated as $\hbar(m)$.
Suppose that there exists a series of linear maps $\{ \phi_m :  B^m \to A \}$ 
and a series of basis of $B^m$ 
satisfying the following conditions.
\begin{enumerate}
\item $\{ \phi_m  (X_i ) | \mbox{$X_i$ is in the basis of $B^m$} \}$ is linearly independent, and each $\phi_m (X_i )$ is independent of $\hbar (m )$.
\item 
As $\dim V^m \rightarrow \infty $, 
the number of pairs of basis elements
$X_i, X_j \in B^m $ satisfying the following condition tends to infinity.
$$\phi_m (X_i X_j) =\phi_m (X_i ) \phi_m (X_j) + \tilde{O} (\hbar(m)).$$
\end{enumerate}
In this case, we say that $\phi_m$ is an asymptotic algebra homomorphism.
\end{definition}
The basic idea underlying this definition and the subsequent discussions is presented in \cite{fuzzy1}.
Condition 1 is imposed to eliminate trivial mappings 
$\phi_m$, such as the zero map.
\\

The difficulty in showing that
$\phi_\mu$, as defined by (\ref{phi_mu}), is
an asymptotic algebra homomorphism 
lies in the fact that, in the case of matrix representation, 
the symmetrized $E_i E_j$ is not always included in a basis.
To observe the property of asymptotic homomorphism of $\phi_\mu$, we further restrict the construction method of a basis 
$E_i (i=1,2, \cdots , D)$ as follows.
We shall name index multisets as 
$I^1_j = \{\!\!\{ j \}\!\!\}, I^2_j=\{\!\!\{  j_1 , j_2 \}\!\!\} , \cdots , I^k_j=\{\!\!\{ j_1 , \cdots ,  j_k \}\!\!\}$ and denote 
$e^{(\mu)}_{(j_1, \cdots , j_k)} $ by 
$e^{(\mu)}_{I^k_j} $. 
Note that a multiset differs from a set in that it distinguishes the degree of overlap of its elements.

A notation such as $I^k_j \sqcup  I^l_i :=\{\!\!\{ j_1 , \cdots ,  j_k ,  i_1 , \cdots ,  i_l \}\!\!\}$ is also used. 
\begin{enumerate}
\item We choose the unit matrix $Id $ and $E_{I_1^1}=e^{(\mu )}_1, E_{I_2^1}=e^{(\mu )}_2 , \cdots , 
E_{I^1_d} =e^{(\mu )}_d$ as a part of the basis.
For the sake of uniform description, we denote these as $E_{I^0}:= Id ~(d_0:=1) $ and $E_{I^1_i} (i=1,\cdots , d=:d_1)$.
\item From $e^{(\mu)}_{(1,1)},  e^{(\mu)}_{(1,2)}, \cdots , e^{(\mu)}_{(1,d)}, e^{(\mu)}_{(2,2)} ,e^{(\mu)}_{(2,3)} ,\cdots, e^{(\mu)}_{(2,d)},
\cdots , e^{(\mu)}_{(d,d)} $ ,
the maximum linearly independent terms of these
are selected and added to $Id=: E_{I^0} ~(d_0:=1) $ and $E_{I^1_i} (i=1,\cdots , d_1)$
to be used as the basis. Let $d_2$ be the number of them, and 
$E_{I^2_j} (j=1,\cdots , d_2)$ be the element of the basis. 
The index multiset ${I^2_j} (j=1,\cdots , d_2)$ is also fixed.
\item 
We choose the largest number of linearly independent elements from  $e^{(\mu )}_{I^1_i \sqcup I^2_j} 
(1\le i \le d_1, 1 \le j \le d_2)$, and 
if more linearly independent elements can be made from 
$e^{(\mu )}_{I^1_i \sqcup I^1_j  \sqcup I^1_k} 
(1\le i , j , k \le d_1)$, they are also added to the basis components.
We denote them as $E_{I^3_j} (j=1,\cdots , d_3)$.
Here $d_3$ is the maximum number of linear independent terms.
The fixed index multiset ${I^3_j} (j=1,\cdots , d_3)$
are elements of 
$\{ I^1_i \sqcup I^1_j  \sqcup I^1_k ~ |~ 1\le i,j,k \le d \}$.
\item Using $e^{(\mu )}_{I^2_k \sqcup I^2_j} 
(1 \le j , k\le d_2)$ and $e^{(\mu )}_{I^1_i \sqcup I^3_j} 
(1\le i \le d_1, 1 \le j \le d_3)$, 
we do the similar process.
If necessary, $e^{(\mu )}_{I^1_i \sqcup I^1_j  \sqcup I^1_k  \sqcup I^1_l} 
(1\le i , j , k , l\le d_1)$ is also taken into account in the 
elements of the basis as well.
$E_{I^4_j} (j=1,\cdots , d_4)$ are chosen. 
\item We repeat the same process as above until the basis of $T_\mu$ is completed. $d_0 + d_1 +\cdots + d_{n_\mu} = D$.
\end{enumerate}

The basis of $T_\mu$ can be constructed in this way. 
In other words, it can be constructed so that 
the index multiset that distinguishes the elements of a basis 
is the union of the index multisets of the other elements. 
This is the point of this method of constructing a basis.

The following asymptotic algebra homomorphisms are established in the correspondence of the multisets of their indices.
\begin{proposition}\label{prop_Kinji}
Let $E_{I^k_i} = e^{(\mu)}_{(i_1 , \cdots , i_k)}, E_{I^l_j} = e^{(\mu)}_{(j_1, \cdots , j_l)}
$ and $E_{I^k_i \sqcup  I^l_j }= e^{(\mu)}_{(i_1 , \cdots , i_k , j_1, \cdots , j_l)}$ be elements of basis of $ T_\mu$. 
Then the following is obtained.
\begin{align}
\phi_{\mu} ( E_{I^k_i} ) \phi_{\mu} ( E_{I^l_j}  ) = \phi_{\mu} (E_{I^k_i}  E_{I^l_j} ) +  \tilde{O}(\hbar({\mu}) ).
\label{3_12}
\end{align}
\end{proposition}

\begin{proof}
The left-hand side of (\ref{3_12}) is given as
\begin{align*}
\phi_{\mu} ( E_{I^k_i} ) \phi_{\mu} ( E_{I^l_j}  ) 
&=    x_{i_1} \cdots x_{ i_k} x_{j_1} \cdots x_{ j_l}. 
\end{align*}
On the other hand, by using (\ref{hbarCommRel}) $E_{I^k_i}  E_{I^l_j}$ is expressed as
\begin{align}
 E_{I^k_i}  E_{I^l_j}
 =
  e^{(\mu)}_{(i_1, \cdots i_k ,j_1, \cdots ,j_{l})} 
 -\sum_{m < k+l} \sum_{n=1}^{d_m}  \hbar(\mu)^{k+l-m} c_{I^m_n,(i_1, \cdots ,j_{l})}  E_{I^m_n}
 , 
\end{align}
where $ c_{I^m_n,(i_1, \cdots ,j_{l})} $ is a complex number that does not depend on $\hbar (\mu ) $.
Because 
$E_{I^k_i \sqcup I^l_j }= e^{(\mu)}_{(i_1 , \cdots , i_k , j_1, \cdots , j_l)}$
is contained in the basis, (\ref{3_12})  leads to the following equation;
\begin{align*}
\phi_{\mu} (  E_{I^k_i}  E_{I^l_j} )
&= 
\phi_\mu ( e^{(\mu)}_{(i_1, \cdots i_k ,j_1, \cdots ,j_{l})} - \sum_{m < k+l} \sum_{n=1}^{d_m}  \hbar(\mu)^{k+l-m} c_{I^m_n,(i_1, \cdots ,j_{l})}  E_{I^m_n} )\\
&= 
x_{i_1} \cdots x_{ i_k} x_{j_1} \cdots x_{ i_l}
 - \sum_{m < k+l} \sum_{n=1}^{d_m}  \hbar(\mu)^{k+l-m} c_{I^m_n,(i_1, \cdots ,j_{l})} x_{n_1} \cdots x_{n_m}
 .
\end{align*}
Since $\sum_{m < k+l} \sum_{n=1}^{d_m} $ is a finite sum, 
the following conclusion is obtained.
\begin{align}
\phi_{\mu} ( E_{I^k_i} ) \phi_{\mu} ( E_{I^l_j}  ) = \phi_{\mu} (E_{I^k_i}  E_{I^l_j} ) +  \tilde{O}(\hbar({\mu}) ).
\end{align}
\end{proof}

More generally, this proposition can be generalized as follows.
\begin{proposition}\label{prop_Kinji_general}
Let $E_{I^{p_1}_i}, E_{I^{p_2}_j}, \cdots, E_{I^{p_m}_k}
$ and $E_{I^{p_1}_i \sqcup \cdots  \sqcup I^{p_m}_k }$ be elements of basis of $ T_\mu$.
Then the following is obtained
\begin{align}
\phi_{\mu} ( E_{I^{p_1}_i} ) \cdots \phi_{\mu} ( E_{I^{p_m}_k}  ) 
= \phi_{\mu} (E_{I^{p_1}_i}  \cdots  E_{I^{p_m}_k} ) +  \tilde{O}(\hbar({\mu}) ) .
\end{align}
\end{proposition}
The proof is obtained in the same way as the proof of Proposition \ref{prop_Kinji}.\\

\begin{rem}
Consider the case where the dimension ${\rm dim} V^{\mu}$ of the representation space $V^\mu$ increases.
Suppose that the representation $\rho^\mu$ is an irreducible representation. 
In this case, the dimension of the algebra $T_\mu$ generated by the basis of the Lie algebra, 
i.e., the number of $E_1, ..., E_D$, increases.
Therefore, the number of pairs of elements of the basis that satisfy the above Proposition \ref{prop_Kinji} also increases.
In the limit ${\rm dim} V^{\mu} \rightarrow \infty$, we can understand $\mathbb{C}[x]$ as an approximation to
the algebra $T_\mu$ in the sense that the number of pairs of elements satisfying 
the above Proposition \ref{prop_Kinji} increases infinitely.
\end{rem}

\begin{rem}
Using this $\phi_\mu$, we can write the necessary condition that
a sequence of quantizations 
$\{ q_\mu : A_{\mathfrak{g}} \to T_\mu \subset End (V^\mu)\}$ is a weak matrix regularization
as
$$\phi_\mu ([q_{\mu} ( f ) , q_{\mu} ( g )])
= \hbar ({\mu}) \phi_\mu (q_{\mu} ( \{ f , g \} ) )+ \tilde{O}(\hbar^2(\mu)) $$
for any $f, g \in A_{\mathfrak{g}}$.
The existence of this asymptotic algebra homomorphism map 
$\phi_\mu$ clarifies that the quantization correctly reflects noncommutativity.
There exists an $\hbar (\mu)$ in the image of $q_\mu$ that is derived from the degree of the polynomial 
in $ A_{\mathfrak{g}}$.
On the other hand, in the image of this $\phi_\mu$, only the $\hbar (\mu)$ 
derived from commutators or, due to the finite dimensionality of the matrix, the 
$\hbar (\mu)$ that appears only from higher-degree terms exceeding the matrix dimension 
(i.e., the $\hbar (\mu)$ caused by Lemma \ref{lemma3_2}) appears.
The $\hbar (\mu)$ that appears in Lemma 3.2 originates from higher-order terms corresponding to monomials that are mapped outside the adaptive part of $gl(V^\mu)$.
Therefore, such terms are ignored in the commutative 
limit as well as higher-order of noncommutativity.
\end{rem}

\bigskip

Since the quantization by $q_\mu$ is defined as $q_\mu (x_i ) = e^{(\mu )}_i = \hbar (\mu ) \rho^\mu (e_i)$, 
the degree of $\hbar$ changes depending on the degree of the monomial.
As already pointed out in the case of the Fuzzy sphere by Chu, Madore, and Steinacker \cite{Chu:2001xi}, 
the classical (commutative) limit of matrix regularization changes 
its result by fixing the ratio of $1/\hbar$ and $\dim V^\mu$.
The Casimir operator is a quantity that has one typical feature that is consistent with a Lie algebra and a corresponding Lie-Poisson algebra.\\

The eigenvalues of the Casimir operator are determined by the representation.
When the dimension of the representation space becomes large, 
in the case of the fuzzy sphere,
the eigenvalues of the Casimir operator also increase, and one reduces 
$\hbar$ so as to counterbalance this.
In other words, the property that the limit in which the dimension of the representation space diverges
yields a commutative limit is an expected feature in ordinary matrix models.
This property of yielding a commutative limit is an expected feature in ordinary matrix models.
Although this requirement is not necessary merely to satisfy the definition of quantization,
we will nevertheless make reference to it here.
Assume that there are 
$d_c$ independent $k$th-degree Casimir operators $C_{k i} (i=1,2, \cdots , d_c)$.
We require that the Lie algebra admit a sequence of irreducible representations 
$\{ V^\mu \}$ whose dimension increases,
and that the eigenvalues of the 
$k$th-degree Casimir operator diverge; in other words, that the following condition be satisfied.
\begin{align}
C_{k i} = \Lambda^k_i (V^\mu) Id_\mu , \quad
\lim_{\dim V^\mu \to \infty} |\Lambda^k_i (V^\mu) | = \infty ,
\quad (i=1,2, \cdots , d_c),
\label{okubo} 
\end{align}
where $Id_\mu$ is the unit matrix in $gl(V\mu)$ and
$\Lambda^k_i (V^\mu) $ is the eigenvalue of $C_{ki}$.
This condition is naturally satisfied when the Lie algebra is semisimple.
This formula corresponds to the eigenvalue of Casimir operator
that does not depend on $\hbar$.
Casimir polynomials of $A_\mathfrak{g}$ are defined by
\begin{align}
\{ x_i , f(x) \} = 0~ ( i = 1, \cdots , d ) ,
\end{align}
and we denote the set of all  Casimir polynomials of $A_\mathfrak{g}$ by $CaP$.
\begin{proposition}
For any $k$th-degree Casimir polynomial 
$ f^C_k  \in CaP$, $\hbar^k (\mu )C_k := q_\mu  (f^C_k ) \in T_\mu$ 
 is a Casimir operator.
\end{proposition}
This has already been given in \cite{AA}, albeit for enveloping algebras.
Also, the proof of the first half of Proposition \ref{prop_qu} corresponds to that proof, 
so you may refer to it there.
Since (\ref{okubo}), 
$\lim_{\dim V^\mu \rightarrow \infty} |C_{k} | = \infty$. 
So, if $\hbar$ is chosen such that the absolute value of the eigenvalue of $q_\mu  (f^C_k ) $
\begin{align}
{ |C_k| }{|\hbar (\mu ) |^{k}}  \label{hbar_ratio}
\end{align}
is fixed,
then the correspondent $k$th-degree Casimir polynomial 
survives under the
$\hbar(\mu) \rightarrow 0 , {\dim} V^\mu \rightarrow \infty$.
For later convenience, we will uniquely determine $\hbar(\mu)$ as follows.
Let $C^k_i(x) = \sum_J C_{i, J} x^{J}= \sum_J C_{i , J} x_{j_1}\cdots x_{j_k} (i=1,2, \cdots , d_c)$ be linearly independent  $k$th-degree Casimir polynomials.
We define $C^k_i( e^\mu) $ by
\begin{align}
C^k_i( e^\mu) := q_\mu ( C^k_i(x) )
= \sum_J C_{i ,J}  e^{(\mu)}_{(j_{1} , \cdots , j_k)}
=
\frac{\hbar^k (\mu )}{k!}
\sum_J C_{i , J} 
\sum_{\sigma \in Sym(k)}
 \rho^\mu (e_{j_{\sigma(1)}}) \cdots \rho^\mu (e_{j_{\sigma(k)}}).
\end{align}
From Schur's lemma, this Casimir operator is proportional to the unit matrix for semisimple Lie algebras.
Condition (\ref{okubo}) is imposed to account for the possibility 
that the Lie algebra is not semisimple.
Under  (\ref{okubo}), it can be fixed to any one eigenvalue $\lambda^k_i \in {\mathbb C}$ 
of the matrix $C^k_i( e^\mu)$
by determining the sequence of $\hbar(\mu)$ 
and $V^\mu$ appropriately;
\begin{align} \label{fixed_casimir_relation}
C^k_i( e^\mu) = \hbar^k (\mu ) C_{k i}
= \lambda^k_i ,
\end{align}
for any $q_\mu : A_\mathfrak{g} \to V^\mu $.
This $\lambda^k_i$ depends on $ \hbar (\mu )$
and $\dim V^\mu$, but it does not depend on $\mu$ by this definition.
As such, a single sequence $\hbar(\mu)$, 
defined so as to fix a given eigenvalue 
may also simultaneously determine multiple eigenvalues. 
We denote the number of such eigenvalues by $l$, 
and in what follows, we consider the set of Casimir operators 
$$
\{ C^{k}_1(e^{(\mu)}),  C^{k}_2(e^{(\mu)}), \cdots , C^{k}_l(e^{(\mu)}) \}
$$
whose eigenvalues are fixed simultaneously, as well as the corresponding set of Casimir polynomials
$$
\{ C^{k}_1(x),  C^{k}_2(x), \cdots , C^{k}_l(x) \}.
$$
The Casimir operators whose eigenvalues are fixed simultaneously 
are not necessarily limited to those of degree $k$.
Therefore, it might be more appropriate to denote by $C^{k(i)}_i (i=1,2, \cdots , l)$ 
as the Casimir operator of degree $k(i)$ above. 
However, for simplicity, we will denote them by 
$C^{k}_i$ in what follows.
This fixed sequence $\hbar(\mu)$  implies that equation
(\ref{fixed_casimir_relation}) approaches 
$C^k_i(x)= \lambda^k_i ~(i=1,2,\cdots,l)$ in the limit
$\hbar(\mu) \rightarrow 0 ,~ {\dim} V^\mu \rightarrow \infty$.
This is the reason for describing varieties in this paper 
in terms of Casimir polynomials
of degree $k$.\\

We comment on this eigenvalue (\ref{fixed_casimir_relation}).
By ``the eigenvalue is fixed'', we mean that it is the same value for any 
$q_\mu$ in the series of matrix regularization $\{q_\mu\}$.
$\lambda^k_i $ is still a polynomial of degree $k$ in $\hbar(\mu)$ 
and is still an $\tilde{O}(\hbar^k(\mu))$.\\

\section{ (Weak) Matrix regularization for Lie-Poisson varieties} \label{sect4}

A Casimir polynomial $f(x) \in A_\mathfrak{g}$ is defined by
$
\{ x_i , f(x) \} = 0~ ( i = 1, \cdots , d ) ,
$
and we denote the set of all  Casimir polynomials of $A_\mathfrak{g}$ by $CaP$.
Consider a vector space $C \subset CaP$ whose basis is $\{ f_i^C \} $, i.e.,
$C:= \{ \sum_i a_i f_i^C ~|~ a_i \in  \mathbb{C},  ~f_i^C \in CaP \}$.
We introduce an ideal of $A_\mathfrak{g} $ generated by $C$ as
\begin{align}
I(C) := \left\{ \sum f_i^C(x) g_i(x) \in  A_\mathfrak{g} ~ | ~ 
f_i^C(x) \in C, ~
g_i(x) \in  A_\mathfrak{g} ~ \right\} .
\end{align}
This ideal is compatible with the Poisson structure because $\{ x_i , f_j^C(x) \}=0$.
So, we can introduce the new Poisson algebra as follows:
\begin{align}
A_\mathfrak{g}  / I(C) := \left\{ [f(x)]  ~ \mid ~ f(x) \in A_\mathfrak{g} 
\right\} ,
\end{align}
where $ [f(x)] = \{ f(x) + h(x) ~\mid h(x) \in I(C) \} $, and the sum and multiplication are defined
as $[f(x)]+[g(x)]= [f(x)+g(x)]$ and $[f(x)]\cdot [g(x)] = [f(x)\cdot g(x)]$.
The Poisson bracket is also defined by 
as $$ \{ [f(x)] , [g(x)] \} := [ \{ f(x) , g(x) \}] .$$
We abbreviate this Poisson algebra $(A_\mathfrak{g}  / I(C) , \cdot  , \{ ~,~ \})$ as $A_\mathfrak{g}  / I(C)$.
In this section, we formulate the quantization of this Poisson algebra by means of weak matrix regularization.

\begin{rem}
In this paper, we call $A_\mathfrak{g}  / I(C)$ itself or 
the variety defined by $I(C)$  ``Lie-Poisson variety''.
However, it does not mean $I(C)$ is a prime ideal.
We use the term variety to mean an algebraic variety defined by $f_i^C(x) =0$.
\end{rem}
\bigskip


The quotient and remainder of a multivariate polynomial cannot be uniquely determined in general. However, after choosing a monomial ordering, such as the lexicographic order, and thereby inducing an ordering on multivariate polynomials, the remainder of a polynomial can be uniquely defined.

It is known that if we fix the ordering every ideal $I$ has a unique reduced Gr\"obner basis,
and  the following fact is known. 
(See the Appendix  \ref{appendix_grobner} for definitions and necessary information on the Gr\"obner basis.).
\begin{theorem}\label{reducedGrobner} (See for example \cite{Dumniit_Foote_Abstract Algebra,{Cox_Little_Oshea}}.)
Fix a monomial ordering on $K [x_1, \dots , x_n]$.
\begin{enumerate}
\item Every ideal $I \subset  K [x_1, \dots , x_n]$ has a unique reduced Gr\"obner basis.
\item  Let ${g_1, \dots , g_m}$ be the Gr\"obner 
basis for the ideal $I$ in $K [x_1, \dots , x_n]$. 
Then every polynomial $f \in K [x_1, \dots , x_n]$ can be written uniquely in the form 
\begin{align}
f= h + r \label{IplusG}
\end{align}
where $h  \in  I$ and no monomial term of the $r$ is divisible by any $LT(g_i)$.
\end{enumerate}

\end{theorem}

In the following, we fix a monomial ordering by the graded lexicographic ordering.


\subsection{Formulation of quantization via enveloping algebra} \label{4_1}
In this subsection, we formulate a quantization map from  $A_\mathfrak{g}  / I(C) $ to
a quotient algebra of the universal enveloping algebra 
by a two-sided ideal.\\

We use $\mathbb{C}[\hbar ]$ as the commutative ring $R$, where $\hbar$ is a complex variable in this subsection.
%
As in Subsection \ref{3_2}, we use the 
enveloping algebra $ \mathcal{U}_\mathfrak{g} [ \hbar ]$.
Recall the two-sided ideal is $I=\{ \sum_{i,j,k,l}  a_{k} ( X_i   X_j  - X_j   X_i  -  [X_i, X_j ] ) b_{l} ~ | ~ a_{k} , b_{l} \in R \langle X \rangle  \}$,
and $ \mathcal{U}_\mathfrak{g} [ \hbar ]$ is defined by
$\mathcal{U}_\mathfrak{g}[\hbar ]
= R \langle X \rangle  / I .$
%
\\

We comment on why we need to introduce 
$\mathcal{U}_\mathfrak{g}[\hbar ]$ here.
The deformation using the independent parameter $\hbar$ as above means that the generator 
$X_i$ is completely independent of $\hbar$.
Fuzzy space or matrix regularization, $q_\mu$, is an expansion with 
respect to $\hbar$, like a Taylor expansion. 
Care is needed to distinguish the degree of $\hbar$, 
which represents noncommutativity, 
and this requires an asymptotic algebraic homomorphism.
On the other hand, the quantization developed in this subsection has a major difference 
in that the degree of the generators and the degree of 
$\hbar$ are not related.\\

We are dealing with what has long been known as the canonical
 mapping from $A_{\mathfrak{g}}$ to $ \mathcal{U}_\mathfrak{g} $
in the studies by Abellanas-Alonso \cite{AA}, (see also \cite{Dixmier}) 
 , and we replace its target space from $ \mathcal{U}_\mathfrak{g} $ to $ \mathcal{U}_\mathfrak{g}  [\hbar]$.
\begin{definition} [Dixmier, Abellanas-Alonso]\label{def_q_u_1}
We define a canonical linear map 
$q_U : A_\mathfrak{g} \to \mathcal{U}_\mathfrak{g} [\hbar]$
by 
\begin{align}
x_{i_1} \cdots x_{i_k} \mapsto \frac{1}{k!}\sum_{\sigma \in Sym(k)} X_{i_{\sigma (1)}} \cdots X_{i_{\sigma (k)}} .
\label{AAD}
\end{align}
\end{definition}

\begin{proposition} \label{prop_qu}
$q_U : A_\mathfrak{g} \to \mathcal{U}_\mathfrak{g}[\hbar] $ is a quantization i.e.,
\begin{align*}
[q_U ( f ) , q_U ( g )]
= \hbar ~ q_{U} ( \{ f , g \} ) + \tilde{O}(\hbar^2)
\end{align*}
for $f,g \in A_\mathfrak{g} $. In other words, $q_U \in Q$.
Especially, if $\min \{ \deg f, \deg g \} \le 1 $, then 
\begin{align}
[q_U ( f ) , q_U ( g )]
= \hbar ~ q_{U} ( \{ f , g \} ) . \label{prop_deg1}
\end{align}
\end{proposition}

\begin{proof}
Since there is an anti-symmetry with respect to the permutation of $f$ and $g$, it is enough that we write the proof for $ \deg f \le \deg g$.
In addition, since it is sufficient to prove the case of monomials by using the linearity of $q_U$, Poisson brackets, and Lie brackets, 
we treat $f$ and $g$ as monomials below.
The case of $\deg f = 0$ is trivial.
Let us consider the case of  $\deg f = 1$.
For monomials $f(x)= x_i $ and $g = x^{I^m} = x_{i_1} \cdots x_{i_m} $ ,
\begin{align*}
 q_{U}( \{ x_i, x^{I^m} \} ) &=  \sum_{l, k} f^k_{i l} q_{U} ( x_k \partial_{l} x^{I^m} )
\\
&=  \sum_{l, k} f^k_{il}~ q_{U}\left( x_k \sum_{n=1}^{m} x_{i_1} \cdots x_{i_{n-1}}  
\delta_{i_n l} x_{i_{n+1}} \cdots x_{i_m} \right) .
\end{align*}
Let us introduce the indices $(j_1^{(n)} , j_2^{(n)}, \cdots ,  j_m^{(n)}) := ( i_1 , \cdots , i_{n-1} , k , i_{n+1} , \cdots , i_m )$.
Then, we obtain
\[
q_{U}( \{ x_i, x^{I^m} \} ) = \frac{1}{m!} \sum_{\sigma \in Sym(m)} \sum_{n=1}^{m} \sum_{k=1}^d
 \left( X_{j_{\sigma(1)}^{(n)}} \cdots (  f^k_{i i_n} X_{j_{\sigma(\sigma^{-1}(n))}^{(n)}} ) \cdots X_{j_{\sigma(m)}^{(n)}} \right).
\]
Here $ f^k_{i i_n} X_{j_{\sigma(\sigma^{-1}(n))}^{(n)}} \!\!\!\!\!\!\!
= f^k_{i i_n} X_{j_{n}^{(n)}} 
=  f^k_{i i_n} X_{k}$ 
is the $ \sigma^{-1}(n) $-th item from the left. Finally, we find 
\begin{align*}
\hbar~ q_{U}( \{ x_i, x^{I^m} \} ) &=
\frac{1}{m!} \sum_{\sigma \in Sym(m)} \sum_{n=1}^{m} X_{i_{\sigma(1)}} \cdots X_{i_{\sigma(\sigma^{-1}(n)-1)}} [ X_i, X_{i_{n}} ] X_{i_{\sigma(\sigma^{-1}(n)+1)}} \cdots X_{i_{\sigma(m)}} \\
&= \frac{1}{m!} \sum_{\sigma \in Sym(m)} \sum_{n=1}^{m} X_{i_{\sigma(1)}} \cdots X_{i_{\sigma(n-1)}} [ X_i, X_{i_{\sigma(n)}} ] X_{i_{\sigma(n+1)}} \cdots X_{i_{\sigma(m)}}\\
&= \frac{1}{m!} \sum_{\sigma \in Sym(m)} [ X_i, X_{i_{\sigma(1)}} \cdots X_{i_{\sigma(m)}} ]
=[ q_{U} (x_i), q_{U} (x^{I^m}) ] .
\end{align*}
Therefore, (\ref{prop_deg1}) is shown for 
$\min \{ \deg f, \deg g \} \le 1 $.\\

Next, we consider the case with $\deg f \ge 2$. 
Essentially, the case $\deg f \ge 3$ is the same as the case $\deg f = 2$,
so for simplicity we only describe the case where $\deg f =2$.
For monomials $f(x)= x_i x_j$ and $g = x^{I^m} = x_{i_1} \cdots x_{i_m} $,
\begin{align}
q_{U}(\{ x_i x_j, g \}) = q_{U} (x_i \{ x_j, g \} + x_j \{ x_i, g \}) = 
\sum_{l,k} \left\{ q_{U} (x_i f^k_{j l} x_k \partial_l x^{I^m}) + q_{U} (x_j f^k_{i l} x_k \partial_l x^{I^m})
\right\}.
\label{m2case1}
\end{align}
Let us take a closer look at the first term $q_{U} (x_i f^k_{j l} x_k \partial_l x^{I^m})$. 
After the similar calculations as in the case of $\deg f = 1$, 
using 
the indices $(\alpha_1 ,\alpha_2 , \cdots ,  \alpha_{m+1}) := 
(i_1 , \cdots , i_{n-1} , k , i_{n+1} , \cdots , i_m , i)$,
the following is obtained.
\begin{align*}
&\sum_{l,k} \hbar q_{U} (x_i f^k_{j l} x_k \partial_l x^{I^m}) = \\
&\frac{1}{(m+1)!} \sum_{\sigma \in Sym(m+1)} \sum_{n=1}^{m} X_{\alpha_{\sigma(1)}} \cdots X_{\alpha_{\sigma(\sigma^{-1}(n)-1)}} 
 [ X_j, X_{i_n}] X_{\alpha_{\sigma(\sigma^{-1}(n) +1)}} \cdots X_i \cdots X_{\alpha_{\sigma(m+1)}}.
\end{align*}
Here, $X_i$ is located at the $\sigma^{-1}(m+1)$-th position from the left.
Using the commutation relation, 
we can write the above with $X_i$ at the beginning and at the end as follows.
\begin{align}
\sum_{l,k} \hbar q_{U} (x_i f^k_{j l} x_k \partial_l x^{I^m}) =&
 \frac{1}{2} \frac{m+1}{(m+1)!} X_i \!\!\!\! \sum_{\sigma \in Sym(m)} 
 \!\! \sum_n X_{i_{\sigma(1)}} \cdots  [ X_j , X_{i_{n}} ] \cdots X_{i_{\sigma (m)}}  \label{m2case2}\\ 
+&\frac{1}{2} \frac{m+1}{(m+1)!} \!\!\!\! \sum_{\sigma \in Sym(m)} 
\!\! \sum_n X_{i_{\sigma(1)}} \cdots  [ X_j , X_{i_{n}} ] \cdots X_{i_{\sigma (m)}} X_i +\tilde{O}(\hbar^2 ).
\notag
\end{align}
The second term in (\ref{m2case1}) is obtained by swapping the $i$ and $j$ in (\ref{m2case2}).
Therefore, 
\begin{align}
\hbar~  q_{U}(\{ x_i x_j, g \}) =
&
 \frac{1}{2} \left( X_i  [ X_j , q_{U} (x^{I^m}) ] 
+  [ X_j , q_{U} (x^{I^m}) ]  X_i  +
   X_j  [ X_i , q_{U} (x^{I^m}) ]+
  [ X_i , q_{U} (x^{I^m}) ]  X_j \right) \notag\\
&+\tilde{O}(\hbar^2 ).
\label{m2case3}
\end{align}
On the other hand, 
\begin{align}
&[ q_{U} (x_i x_j), q_{U} (x^{I^m}) ] = \frac{1}{2} [ X_i X_j + X_j X_i,  q_{U} (x^{I^m}) ] \label{m2case4}\\
&= \frac{1}{2} \left( X_i  [ X_j , q_{U} (x^{I^m}) ] 
+  [ X_j , q_{U} (x^{I^m}) ]  X_i  +
   X_j  [ X_i , q_{U} (x^{I^m}) ]+
  [ X_i , q_{U} (x^{I^m}) ]  X_j \right). \notag
\end{align}
From (\ref{m2case3}) and (\ref{m2case4}), the desired result is proven.
\end{proof}

Next, in order to construct the quantization of $A_\mathfrak{g} / I(C) $, the following ideal 
$I(C(X)) \subset \mathcal{U}_\mathfrak{g}[\hbar] $
is also introduced into the enveloping algebra.
\begin{align}
I(C(X))&:= 
 \left\{ \sum_{i,j,k} a_j (X) f_i^C(X) b_k (X) \in \mathcal{U}_\mathfrak{g}[\hbar] ~ | ~ 
f_i^C(x) \in CaP, ~
a_j(X), b_k(X) \in  \mathcal{U}_\mathfrak{g}[\hbar] ~ \right\} .
\end{align}
Here, we use
\begin{align}
f_i^C(X) : = q_U ( f_i^C(x) ) , \quad   \quad 
f_i^C (x) \in CaP := \{ f(x) \in  A_\mathfrak{g} ~| ~ \{x_i , f(x) \}=0 (i =1, \cdots , d ) \}.
\end{align}
Note that from (\ref{prop_deg1}), 
\begin{align}
[ X_j ,  f_i^C(X) ] = [X_j ,  q_U ( f_i^C(x) )  ] = 0 , \quad ( j =1, \cdots , d ). \label{casimir_u}
\end{align}
In short $f_i^C(X) $ is a Casimir operator in $\mathcal{U}_\mathfrak{g}[\hbar] $.
Therefore, it is possible to define $\mathcal{U}_\mathfrak{g}[\hbar] / I(C(X)) $
while keeping it compatible with the commutation relations.
We denote $\{ f(X) + h(X) ~|~ f(X) \in \mathcal{U}_\mathfrak{g}[\hbar] , h(X) \in  I(C(X)) \}$
by $[f(X)]$. 
The sum and product of the algebra  $\mathcal{U}_\mathfrak{g}[\hbar] / I(C(X)) $ are defined 
by $[f]+[g]:= [f+g]$ and $[f][g]=[f g]$, and the commutator product is determined by 
$[[f] , [g] ] = [ [f , g ] ]$.
Let us generalize $q_U : A_\mathfrak{g} \to \mathcal{U}_\mathfrak{g}[\hbar] $
to $
q_{U/I} : A_\mathfrak{g} /I(C) \to \mathcal{U}_\mathfrak{g}[\hbar] / I(C(X)) 
$.
\begin{definition} \label{def_q_U_2}
A linear map
\[
q_{U/I}  : A_\mathfrak{g}  /I(C) \to  \mathcal{U}_\mathfrak{g} [\hbar]/I(C(X))
\]
is defined as follows.
Let $G$  be a reduced Gr\"obner basis of $I(C)$.
For any $f(x) \in A_\mathfrak{g} $,  $r_{f, G}$ is uniquely determined 
by Theorem \ref{reducedGrobner} as
\[
f(x) = r_{f, G}(x) + h_f , \quad h_f \in I(C).
\]
For $\forall [f(x)] \in  A_\mathfrak{g}  /I(C)$, we define
\[
q_{U/I}  ([f(x)]) :=  [ q_{U}  ( r_{f, G}(x) ) ],
\]
where
$q_{U}  ( r_{f, G}(x) )$ is determined by Definition \ref{def_q_u_1}.
\end{definition}

\begin{theorem}
The above $
q_{U/I} : A_\mathfrak{g} /I(C) \to \mathcal{U}_\mathfrak{g}[\hbar] / I(C(X)) 
$ is a quantization;
\begin{align}
[ q_{U/I} ( [f] ) ,  q_{U/I} ( [g] ) ] = \hbar   q_{U/I} (  \{ [f ] ,  [g] \}  ) + \tilde{O} (\hbar^2 ) ,    
\end{align}
for $\forall [f] , [g]  \in  A_\mathfrak{g} / I(C)$.
$\tilde{O}(\hbar^{n})$ is used in the sense of Example \ref{exA_4} in Appendix \ref{ap1}.
\end{theorem}

\begin{proof}
In the following, for any polynomial $f$, 
 we write $f: =h_f + r_f $ to mean equation (\ref{IplusG}), i.e.,
$h_f \in I(C), r_f \notin I(C)$.
For any $f,g \in A_\mathfrak{g} $, from
\[
\{f, g\} = h_{\{f, g\}} + r_{\{f, g\}}, 
\]
and
\[
\{f, g\} = \{r_f + h_f, r_g + h_g\} = \{r_f, r_g\} + h ,
\]
where $f=h_f + r_f , ~g = h_g + r_g$ and $h \in I(C)$,
we obtain
\begin{align}
\{r_f, r_g\} = r_{\{f, g\}} + h_{\{r_f, r_g\}}, \quad r_{\{r_f, r_g\}} = r_{\{f, g\}}.
\end{align}
Here the uniqueness described in Theorem \ref{reducedGrobner} 
is used to obtain the above result.
Therefore we find
\begin{align}
\hbar \, q_{U/I} ([\{f, g\}]) = \hbar \,  [q_U(r_{\{f, g\}}) ]
= \hbar  \left[ q_{U} (\{r_f, r_g\})  - q_{U} (h_{\{r_f, r_g\}}) \right]  . \label{proof_AI_quant_1}
\end{align}
From the Proposition \ref{prop_qu}, 
\begin{align}
\left[ \hbar \, q_{U} (\{r_f, r_g\}) \right] &= \left[ [q_{U} (r_f), q_{U} (r_g)] \right] + \tilde{O}(\hbar^2)
\notag \\
&= [ q_{U/I} ([r_f]), q_{U/I} ([r_g])] + \tilde{O}(\hbar^2)
= [ q_{U/I} ([f]), q_{U/I} ([g])] + \tilde{O}(\hbar^2) \label{proof_AI_quant_2}
\end{align}
From (\ref{proof_AI_quant_1}) and (\ref{proof_AI_quant_2}),
this proof is complete if we can show that 
\[
[ q_{U} (h_{\{r_f, r_g\}})] = \tilde{O}(\hbar) .
\]
Recall that $h_{\{r_f, r_g\}}$ is an element of $I(C)$ that is generated by 
Casimir polynomials;
\[
h_{\{r_f, r_g\}} = \sum f_i^C(x) k_i(x), 
\]
where $k_i(x) \in A_{\mathfrak{g}}$ and $f_i^C(x)$ is a Casimir polynomial.
So, it is enough that we show
\[
[ q_{U} (C(x) x^{I}) ] = \tilde{O}(\hbar) ,
\]
for any Casimir polynomial $C(x)$ and any monomial $x^{I}= x_{i_1} \cdots x_{i_m}$.
Recall that $q_{U} (C(x))$ is also a Casimir operator as we saw in (\ref{casimir_u}). So, $q_{U} (C(x)) q_{U} (x^{I})$ is in the ideal $I(C(X))$, and
\begin{align}
[ q_{U} (C(x) x^{I}) ] = [ q_{U} (C(x) x^{I}) - q_{U} (C(x)) q_{U} (x^{I}) ]. \label{proof_AI_quant_3}
\end{align}
We put
$C(x) = \sum_{J} C_{J} x^{J}  ~(C_J \in {\mathbb{C}}, x^{J}= x_{j_1} \cdots x_{j_k})$, then
\[
q_{U} (C(x) x^{I}) = \sum_{J} C_{J} q_{U} (x^{J} x^{I} ).
\]
By the similar discussions with the proof for Lemma \ref{lemma3_4},
$$q_{U} (x^{J} x^{I} )= X_{(j_i , \cdots , j_k , i_1 , \cdots , i_m)}
= X_{(j_i , \cdots , j_k )}X_{(i_1 , \cdots , i_m)}+ \tilde{O}(\hbar)
=q_{U} (x^{J} ) q_U (x^{I} )+ \tilde{O}(\hbar).
$$
Using this, finally (\ref{proof_AI_quant_3}) is written as
\begin{align}
 [ q_{U/I} (C(x) x^{I}) ] &=
 \left[ \sum_{J} C_{J} q_{U} (x^{J} x^{I} ) - q_{U} (C(x)) q_{U} (x^{I}) \right] \notag\\
 &=  \left[ \sum_{J} C_{J} ( q_{U} (x^{J} x^{I} ) - q_{U} (x^{J} ) q_U (x^{I} ) )\right] \notag\\
 &= [ \tilde{O}(\hbar) ] = \tilde{O}(\hbar).
\end{align}
Here $[ \tilde{O}(\hbar) ] = \tilde{O}(\hbar)$ is obtained from Definition \ref{defA2}.

\end{proof}

\begin{rem}
As you can see from the above proof, 
there is no quantization $
A_\mathfrak{g} /I(C) \to \mathcal{U}_\mathfrak{g}[\hbar] 
$ 
in this construction method.
Therefore, it is necessary to discuss the quantization 
$
q_{U/I} : A_\mathfrak{g} /I(C) \to \mathcal{U}_\mathfrak{g}[\hbar] / I(C(X)) .
$ 
\end{rem}


\subsection{(Weak) Matrix regularization for Lie-Poisson varieties $A_\mathfrak{g} /I(C) $}\label{sect4_2}
Up to the previous subsection, we have discussed Casimir polynomials without any conditions.
In this subsection, we will impose restrictions so that it becomes Casimir relations that we need.
This will make it possible to construct matrix regularization of Lie-Poisson varieties.
As already mentioned around (\ref{hbar_ratio}),
by balancing the rate at which 
$\hbar$ approaches zero with the rate at which the dimension of the representation matrix diverges, 
one can ensure that Casimir polynomials of degree 
$k$ define a Lie-Poisson variety in the classical(commutative) limit.
In this following, $k < n_\mu$ is assumed.
\\

We fix $\hbar$ of $ \mathcal{U}_\mathfrak{g}[\hbar] $ as $\hbar = \hbar (\mu)$.
From $ \mathcal{U}_\mathfrak{g}[\hbar] $ to $gl (V^\mu )$ there is a representation $\rho_{U\mu}$ 
as an algebra homomorphism defined by
\begin{align}
\rho_{U\mu} (X_i ) := e^{(\mu)}_i , \quad \rho_{U\mu} (1) := Id^\mu,
\end{align}
where $ Id^\mu$ is the unit matrix.
$e^{(\mu)}$ is the basis of an irreducible representation of $\mathfrak{g}$
introduced in Section \ref{3_2}.
When we fix $\hbar = \hbar (\mu)$,
using this representation we obtain $k$th-degree Casimir operator 
$C^k_i (e^\mu ) := \rho_{U\mu} \circ q_U ( C^k_i (x)) = q_\mu ( C^k_i (x))$.
Here $C^k_i (x)$ is a Casimir polynomial of degree $k$.
It is clear from the definition of $\rho_{U\mu}$ that the commutation relation is unchanged 
even for the image of $\rho_{U\mu}$, so $C^k_i (e^\mu )$ is a Casimir operator. 
If it is an irreducible representation, then $C^k_i (e^\mu )$ is proportional to the unit matrix if it is not zero,
since Schur's lemma.
In the irreducible representation, the Casimir operator 
can be characterized by an eigenvalue.
Let $\lambda^k_i$ denote the eigenvalues of the Casimir operators of $k$-th degree 
$C^k_i (e^\mu ) $.
Recall the discussions at the end of  Subsection \ref{3_2}.
Let us consider the representation for constructing the matrix regularization in that case.
So we chose the generators of the ideal $I(C) \subset A_\mathfrak{g}$ for some fixed $k$ as
\begin{align} \label{4_22}
( f_i^C(x):= C^k_i (x) - \lambda^k_i )_{i \in \{1,\cdots , l\} } ,
\end{align}
where $l \in {\mathbb{N}}$ is chosen as a number of equations 
to determine a Lie-Poisson variety. 
The possible values of $l$ are also restricted by the Lie algebra $\mathfrak{g}$.
(For example, in the case of $\mathfrak{su} (n)$, $l$ is at most $1$.)
Let us reconstruct all in 
Subsection \ref{4_1} using  
\begin{align} \label{4_23}
I(C) := \left\{ \sum f_i^C(x) g_i(x) \in  A_\mathfrak{g} ~ | ~ 
f_i^C(x) =  C^k_i (x) - \lambda^k_i , ~
g_i(x) \in  A_\mathfrak{g} ~ \right\}  \subset A_\mathfrak{g} .
\end{align}
Then, 
\begin{align}
I(C(X))&:= 
 \left\{ \sum a_j (X) f_i^C(X) b_k (X) \in \mathcal{U}_\mathfrak{g}[\hbar] ~ | ~ 
a_j(X), b_k(X) \in  \mathcal{U}_\mathfrak{g}[\hbar] ~ \right\} .
\end{align}
Here, $f_i^C(X)$ is a Casimir operator in $ \mathcal{U}_\mathfrak{g}[\hbar] $ given as
\begin{align}
f_i^C(X) : = q_U ( f_i^C(x) ) = q_U ( C^k_i (x) - \lambda^k_i ).
\end{align}
$ \mathcal{U}_\mathfrak{g}[\hbar] / I(C(X)) ,  q_{U/I}$ and so on are defined by using these ideals.
\begin{definition} A linear function
$\rho_{U/I , \mu} :  \mathcal{U}_\mathfrak{g}[\hbar] / I(C(X))  \to gl (V^\mu )$
is defined as follows. For any monomial $X^I = X_{i_1} \cdots X_{i_m} ~(m \in {\mathbb{N}})$
in $\mathcal{U}_\mathfrak{g}[\hbar]$,
\begin{align}
\rho_{U/I , \mu} ( [X^I ] ) = \rho_{U\mu} (X^I)  = e^{(\mu)}_{i_1} \cdots e^{(\mu)}_{i_m}.
\end{align}
\end{definition}
Let us check the consistency of this definition of $\rho_{U/I , \mu}$.
Let $\Sigma$ be the set of generators of all relations,
\begin{align*}
\Sigma := \left\{
X_j X_k - X_k X_j - [X_j , X_k ] , ~ q_U ( C^k_i (x) ) - \lambda^k_i  ~ | ~ 1\le j,k \le d , i = 1,2, \cdots  , l \right\} .
\end{align*}
We denote the ideal generated by $\Sigma$ by 
$$\displaystyle {\mathfrak I}:= \left\{ \sum_{i,j,k}  a_i  {\mathcal I}_j 
b_k ~ | ~  a_i ,  b_k \in  R \langle X \rangle, ~{\mathcal I}_j \in \Sigma \right\} .$$ 
In other words, $ \mathcal{U}_\mathfrak{g}[\hbar] / I(C(X))  = R \langle X \rangle / {\mathfrak I}$.
For any $[f(X)] = [g(X)]$, there exists $h(X) \in  {\mathfrak I}$ such that
$ g(X) = f(X) + h(X) $. 
Note that 
\begin{align}
\rho_{U/I , \mu} ( [ X_j X_k - X_k X_j - [X_j , X_k ] ]) &=  e^{(\mu)}_j  e^{(\mu)}_k -  e^{(\mu)}_k  e^{(\mu)}_j -
[ e^{(\mu)}_j ,  e^{(\mu)}_k ] = 0  ,   \label{h=0_1}\\
\rho_{U/I , \mu}( [ q_U ( C^k_i (x) ) - \lambda^k_i ]) &= 0 . \label{h=0_2}
\end{align}
These equations mean $h( e^{(\mu)} ) = 0$, so we obtain
\begin{align}
\rho_{U/I , \mu} ( [g(X)] ) = g(e^{(\mu)} ) = f( e^{(\mu)} ) + h( e^{(\mu)} ) = f( e^{(\mu)} ) 
= \rho_{U/I , \mu} ( [f(X)] ) .
\end{align}
In this way, it was confirmed that $\rho_{U/I , \mu} :  \mathcal{U}_\mathfrak{g}[\hbar] / I(C(X))  \to gl (V^\mu )$ is well-defined.
In addition, $\rho_{U/I , \mu}$ is apparently algebra homomorphism.
Linearity is trivial from the definition, and the product is as follows.
\begin{align*}
\rho_{U/I , \mu} ( [f(X)][g(X)] ) = \rho_{U/I , \mu} ( [f(X) g(X)] ) = f( e^{(\mu)} ) g( e^{(\mu)} )
=\rho_{U/I , \mu} ( [f(X)] ) \rho_{U/I , \mu} ( [g(X)] ) .
\end{align*}
\bigskip

\begin{lemma}\label{lem4_7}
For the algebra homomorphism 
$\rho_{U/I , \mu} :  \mathcal{U}_\mathfrak{g}[\hbar] / I(C(X))  \to gl (V^\mu )$,
if $ [f(X)] \in  \mathcal{U}_\mathfrak{g}[\hbar] / I(C(X)) $ is $ \tilde{O} (\hbar^n )$
in the sense of Example \ref{exA_4} in Appendix \ref{ap1}, then
$\rho_{U/I , \mu} ([f(X)])$ is $ \tilde{O} (\hbar^n )$.
\end{lemma}
\begin{proof}
When $ [f(X)]=  \tilde{O} (\hbar^n )$, there exists $g(X) \in \mathcal{U}_\mathfrak{g}[\hbar]$
such that $f(X)= g(X) + I(X) ~ ( I(X) \in {\mathfrak I}) $ and $g(X) = \tilde{O} (\hbar^n )$ by its definition.
$g(X) = \tilde{O} (\hbar^n )$ means that every $a_J (\hbar ) \in {\mathbb C}[\hbar]$ in $g(X) = \sum_{J} a_{J} (\hbar ) X^J $ satisfies
\begin{align*}
\lim_{x\rightarrow 0} \left|
\frac{ a_{J} (x \hbar )}{x^n}
\right| < \infty .
\end{align*}
Recall that $I (e^{(\mu)}) =0$, then
\begin{align}
\rho_{U/I , \mu} ([f(X)]) &= g(e^{(\mu)}) = \sum_{J} a_{J}(\hbar ) e^{(\mu)}_{j_1} \cdots e^{(\mu)}_{j_m}  \notag \\
&= \sum_{J} \hbar^{|J|} a_{J}( \hbar ) \rho^\mu (e_{j_1}) \cdots 
\rho^\mu (e_{j_m}) .
\end{align} 
Here we denote $ |J| :=m$ for $J = (j_1 , j_2 , \cdots, j_m)$.
Then, we find 
\begin{align*}
\rho_{U/I , \mu} ([f(X)]) &= \sum_{J}  \tilde{O} ( \hbar^{n + |J|} ) =   \tilde{O} ( \hbar^n  ) .
\end{align*}
\end{proof}

Recall that $\hbar$ in $ \mathcal{U}_\mathfrak{g}[\hbar] / I(C(X))$ was a variable
introduced independently of $\hbar(\mu)$, at first.
Using the discussions in Subsection \ref{4_1} and this $\rho_{U/I , \mu} $
with Lemma \ref{lem4_7},
we get the following theorem.
\begin{theorem} \label{pre_mat_reg}
The linear map $q^{pre}_\mu : A_\mathfrak{g}  /I(C) \to  gl (V^\mu )$
defined by
\begin{align}
q^{pre}_\mu := \rho_{U/I , \mu} \circ q_{U/I}
\end{align}
is a weak matrix regularization, i.e.,  $\forall [f] , [g]  \in  A_\mathfrak{g} / I(C)$
there exists $P_\mu = \sum_i^D c_i(\hbar(\mu)) E_i  \in T_\mu$, where 
 $c_i(\hbar) $ is a polynomial in $\hbar(\mu)$, such that
\begin{align}
 [ q^{pre}_\mu ( [f] ) ,  q^{pre}_\mu ( [g] ) ] = \hbar   q^{pre}_\mu (  \{ [f ] ,  [g] \}  ) + 
 \hbar^2(\mu) P_\mu .
\end{align}
\end{theorem}
\bigskip

The $q^{pre}_\mu : A_\mathfrak{g}  /I(C) \to  gl (V^\mu )$ introduced above can be
said to be sufficiently weak matrix regularization.
Here, we consider imposing the restriction that the degree of the polynomial 
in the domain of definition, excluding the kernel of the matrix regularization, 
is less than or equal to $n_\mu$. 
This restriction makes $q^{pre}_\mu$
be the generalization of the 
matrix regularization in the case of  
$q_\mu : A_\mathfrak{g}  \to  gl (V^\mu )$ or fuzzy sphere.\\

We introduce a projection map 
$R_\mu : gl (V^\mu )[\hbar(\mu)]  \to gl (V^\mu ) [\hbar(\mu)] $ as follows.
Here $\hbar(\mu)$ is considered to be a variable. 
(Note that $\hbar(\mu)$ was chosen to satisfy  
$ q_U ( C^k_i (x) ) = \lambda^k_i Id $. 
The eigenvalue of the Casimir operator, $\lambda^k_i$, 
can also be freely chosen by the scaling of $\hbar(\mu)$. 
In this sense, we use $\hbar(\mu)$ as a variable.)
Any $M(\hbar(\mu) ) \in gl (V^\mu ) [\hbar(\mu)] $ is expressed as
$M(\hbar(\mu) )= \sum_{0 \le k} \hbar(\mu)^k M_k = \sum_{0 \le k \le n_\mu} \hbar(\mu)^k M_k + \tilde{O}(\hbar(\mu)^{n_\mu +1})$, where each $M_k \in gl(V^\mu)$ does not depend on $\hbar(\mu)$. 
For any $M(\hbar(\mu) )$, we define 
$R_\mu $ by
\begin{align}
R_\mu (M) := \sum_{0 \le k \le n_\mu} \hbar(\mu)^k M_k .
\end{align}
Using this $R_\mu$, let us define a matrix regularization for $A_\mathfrak{g}  /I(C)$.
\begin{definition}
We call 
\begin{align}
q_{A/I ,\mu} : = R_\mu \circ q^{pre}_\mu = R_\mu \circ  \rho_{U/I , \mu} \circ q_{U/I} : A_\mathfrak{g}  /I(C) \to  gl (V^\mu )
\end{align}
a quantization of $A_\mathfrak{g}  /I(C)$.
\end{definition}

\begin{theorem}
$q_{A/I, \mu} $
is a weak matrix regularization, i.e., for $\forall f, g \in A_\mathfrak{g} $, 
there exists $P = \sum_i^D c_i(\hbar(\mu)) E_i  \in T_\mu$,
where each $c_i(\hbar(\mu)) $ is a polynomial in $\hbar(\mu)$, such that
\begin{align}
[ q_{A/I, \mu} ( [ f ] ) ,  q_{A/I, \mu} ( [ g ] ) ] = \hbar(\mu ) q_{A/I, \mu} ( \{ [f] , [g] \} ) + \hbar^2(\mu) P.
\label{q_AI_quantization}
\end{align}
\end{theorem}
\begin{proof}
Since it is linear, it is sufficient to show the case of monomials.
Recall that we fix a monomial ordering for every Lie-Poisson variety by the graded lexicographic ordering.
$\forall f, g \in A_\mathfrak{g} $ with a reduced Gr\"obner basis of $I(C)$
is uniquely expressed as
$f = r_f + h_f, g= r_g + h_g $, where $h_i  \in I(C) $ and
$r_i \notin I(C)$ for any $i \in A_\mathfrak{g}$.
We split $q^{pre}_\mu ( [ r_i ] ) $  as
\begin{align*}
q^{pre}_\mu ( [ r_f ] ) &= F_1 + F_2 ,\\
q^{pre}_\mu ( [ r_g ] ) &= G_1 + G_2 , \\
[F_1 , G_1 ] &= \hbar(\mu)( FG_1 + FG_2 ) ,
\end{align*}
where $\deg F_1 \le n_\mu , \deg G_1 \le n_\mu , \deg FG_1 \le n_\mu$, 
and the degree of any monomial in $F_2$, $G_2$ and $FG_2$ is greater than $n_\mu $.
Using this notation,
\begin{align}
[ q_{A/I, \mu} ( [ f ] ) ,  q_{A/I, \mu} ( [ g ] ) ] 
= [ F_1 , G_1 ] = \hbar(\mu)( FG_1 + FG_2 ). \label{4_36}
\end{align}
On the other hand,
\begin{align}
 \hbar(\mu ) q_{A/I, \mu} ( \{ [f] , [g] \} ) 
 &=  \hbar(\mu ) R_\mu  \circ q^{pre}_\mu ( \{ [r_f] , [r_g] \} ) \notag \\
 &=  \hbar(\mu )  R_\mu  ( \frac{1}{\hbar(\mu)}[ q^{pre}_\mu ( [r_f] ) ,  q^{pre}_\mu ( [r_g] ) ] )
  + \hbar^2 (\mu) P_1 , \label{4_37}
\end{align}
where we use Theorem \ref{pre_mat_reg}.
We denote elements in $T_\mu$ of the same 
type as $P$ by $P_i = \sum_j^D c_j^i (\hbar(\mu)) E_j  $ for $i = 1,2$. 
($P_2$ will be used soon.)
There is a notation of $1/\hbar(\mu)$, but this will not cause any misunderstanding because it cancels out with the $\hbar(\mu)$ 
that arises from the commutator.
The first term in the right-hand side of (\ref{4_37}) is written as
\begin{align}
R_\mu  (\frac{1}{\hbar(\mu)} [ q^{pre}_\mu ( [r_f] ) ,  q^{pre}_\mu ( [r_g] ) ] )
&= R_\mu  \left( \frac{1}{\hbar(\mu)}( [F_1 , G_1 ] + [F_1 , G_2 ] + [ F_2 , G_1 ] + [F_2, G_2] ) \right)
\end{align}
If any of $[F_1 , G_2 ] , [ F_2 , G_1 ] $, or $[F_2, G_2]$ is not $0$,
then its degree is greater than or equal to $n_\mu +2$ and thus
$\displaystyle R_\mu  \big( \frac{1}{\hbar(\mu)}( [F_1 , G_2 ] + [ F_2 , G_1 ] + [F_2, G_2] ) \big) = 0$. 
Here, we used the fact that the commutator product does not change the degree of $\hbar(\mu)$.
Therefore,
\begin{align}
R_\mu  (\frac{1}{\hbar(\mu)} [ q^{pre}_\mu ( [r_f] ) ,  q^{pre}_\mu ( [r_g] ) ] )&= R_\mu  ( \frac{1}{\hbar(\mu)} [F_1 , G_1 ] ) =  FG_1 . \label{4_38}
\end{align}
From (\ref{4_36}), (\ref{4_37}), and (\ref{4_38}),
\begin{align}
[ q_{A/I, \mu} ( [ f ] ) ,  q_{A/I, \mu} ( [ g ] ) ] 
- \hbar(\mu ) q_{A/I, \mu} ( \{ [f] , [g] \} ) 
= \hbar(\mu) FG_2   - \hbar^2(\mu) P_1  . \label{4_39}
\end{align}
Here $FG_2 = \tilde{O} (\hbar^{n_\mu +1}(\mu ) )$.
From Lemma \ref{lemma3_2}, $FG_2$ is expressed as $\hbar(\mu) P_2$. 
Then  we find that (\ref{q_AI_quantization}) is satisfied.
\end{proof}

\bigskip

For $f= h + r_{f,G} $ with
$ \displaystyle
r_{f,G}= \sum_{J} a_{J} x^{J}
= \sum_{\deg x^{J} \le n_\mu } \!\!\!\!\!
a_{J} x^{J}
+
 \sum_{\deg x^{J} > n_\mu  } \!\!\!\!\!
a_{J} x^{J} 
$, the explicit calculation of 
$q_{A/I, \mu}  : A_\mathfrak{g} / I(C)\to gl(V^\mu)$ is given as
\begin{align}
q_{A/I, \mu}  ([f(x)] ) =
\sum_{m \le n_\mu } \!
a_{J}~ \rho_{U/I , \mu}  ([ X_{(j_1 , \cdots  , j_m )} ])
= \sum_{m \le n_\mu } \!
a_{J}~ e_{(j_1 , \cdots  , j_m )}^{(\mu )}.
\end{align}
As can be easily seen from the definition of $q_{A/I, \mu}$, when we chose $I(C) =\{ 0 \} $, 
it is the same as $q_\mu$. 
Therefore, $q_{A/I, \mu}$ is a generalization of $q_\mu$ in Section \ref{3_2}.\\
\bigskip


In the following, we shall consider the approximate homomorphism property once again.
\begin{lemma}\label{lem4_11}
Consider $f= r_f + h_f \in  A_\mathfrak{g}$ with $\deg r_f \le n_\mu$. 
Here $r_f$ and $h_f$ are used in the sense defined in Theorem \ref{reducedGrobner}.
Then we have
$$ q_{A/I ,\mu} ([f]) = q_\mu (r_f) .$$
\end{lemma}
\begin{proof}
By the definition,
$
q_{A/I ,\mu}([f]) = R_\mu \circ  \rho_{U/I , \mu} ([ q_U (r_f) ]) = \rho_{U,\mu} (q_U (r_f)),
$
where we use $\deg r_f \le n_\mu$.
Let us write $r_f = \sum_I a_I x^I$ $( |I| \le n_\mu)$. 
Then we get $q_{A/I ,\mu}([f]) =  \sum_I a_I e^{(\mu)}_{(I)} = q_\mu (r_f).$
\end{proof}

The following proposition is also obtained.
\begin{proposition}\label{prop4_12}
For any $[f], [g] \in A_\mathfrak{g} / I(C)$ with $\deg r_f +\deg r_g \le n_\mu$ and $\deg r_{fg} \le n_\mu$,
there exists $P = \sum_i^{D} c_i(\hbar(\mu)) E_i  \in T_\mu$
with $c_i(\hbar) \in {\mathbb C}[\hbar (\mu )]$ satisfying
$$
q_{A/I, \mu} ( [ f ] )  q_{A/I, \mu} ( [ g ] ) = q_{A/I, \mu} ( [f] [g]) + \hbar(\mu)  P .
$$
Here $r_f$, $r_g$, $h_f$, and $h_g$ are used in the sense defined in Theorem \ref{reducedGrobner}.
\end{proposition}
\begin{proof}
Note that $[fg]= [(r_f +h_f)(r_g+h_g)] = [r_f r_g ] $, $[fg]=[r_{fg} ]$, and $r_f r_g = r_{r_f r_g} + h_{r_f r_g}$.
Then we obtain
\begin{align*}
 [f] [g] = [fg] =  [r_{fg}] =  [r_f r_g ] =  [r_{r_f r_g} ].
\end{align*}
From $\deg r_f +\deg r_g \le n_\mu$ $R_\mu$ acts as an identity,
then
\begin{align*}
q_{A/I, \mu} ([f] [g])  =   q_{A/I, \mu} ([r_{r_f r_g} ]) =  q_\mu  (r_{r_f r_g})=    q_{\mu}  ({r_f r_g}- h_{r_f r_g} )  .
\end{align*}
Here we use Lemma \ref{lem4_11}. 
On the other hand, using  Lemma \ref{lem4_11} and Proposition \ref{prop3_7}, 
there exists $P_1$ such that
$$
q_{A/I, \mu} ( [ f ] )  q_{A/I, \mu} ( [ g ] )= q_\mu ( r_f ) q_\mu (r_g) 
= q_\mu   (r_f r_g )  +  \hbar(\mu)  P_1 .
$$
Therefore, if $q_{\mu}  (h_{r_f r_g} ) = \hbar(\mu) P_2 $ is shown, this proposition is proven.
Here, we denoted elements in $T_\mu$ of the same 
type as $P$ by $P_i = \sum_j^D c_j^i (\hbar(\mu)) E_j  $ for $i = 1,2$.
Note that $h_{r_f r_g} $ is expressed as $\sum f_i^C(x) g_i(x) $ in (\ref{4_23})
and $\deg h_{r_f r_g} \le n_\mu$. So, using Proposition \ref{prop3_7} again, we have
$$
q_{\mu}  (h_{r_f r_g} ) = \sum q_{\mu} ( f_i^C(x) g_i(x) ) = 
\sum q_{\mu} ( f_i^C(x) ) q_{\mu} ( g_i(x) ) +  \hbar(\mu)  P_2.
$$
From the Casimir relation (\ref{fixed_casimir_relation}), i.e.
$q_{\mu} ( f_i^C(x) ) = q_\mu ( C^k_i(x) ) - \lambda^k_i = 0$,
we find that the desired result is obtained.
\end{proof}

\bigskip

Let us show that, 
as in Proposition \ref{prop_Kinji}, there is an asymptotic homomorphism between 
algebra $A_\mathfrak{g} / I(C)$ and algebra $T_\mu$ with a fixed basis $E_1 , \cdots , E_D$.
We define a linear map $[\phi_\mu ] : T_\mu \to  A_\mathfrak{g} / I(C)$ by 
\begin{align}
[\phi_\mu ] (e^{(\mu)}_{(i_1 , \cdots  , i_k ) } ) := [ x_{i_1} \cdots x_{i_k} ]
\end{align}
for $E_I = e^{(\mu)}_{(i_1 , \cdots  , i_k ) } $.

Then the following is obtained.
\begin{proposition}\label{prop_Kinji_2}
Let $E_{I^k_i} = e^{(\mu)}_{(i_1 , \cdots , i_k)}, E_{I^l_j} = e^{(\mu)}_{(j_1, \cdots , j_l)}
$ and $E_{I^k_i \sqcup  I^l_j }= e^{(\mu)}_{(i_1 , \cdots , i_k , j_1, \cdots , j_l)}$ be elements of basis of $ T_\mu$.
Then,
\begin{align}
[\phi_{\mu} ]( E_{I^k_i} ) [\phi_{\mu}] ( E_{I^l_j}  ) = [\phi_{\mu} ](E_{I^k_i}  E_{I^l_j} ) +  \tilde{O}(\hbar({\mu}) ).
\end{align}
\end{proposition}
Since $[x^\alpha][x^\beta]=[x^\alpha x^\beta]$, the proof is the same for Proposition \ref{prop_Kinji}.
Proposition \ref{prop_Kinji_general} can also be generalized to the current case, 
and is obtained by replacing $\phi_\mu$ with $[\phi_\mu ]$.

\begin{proposition}\label{prop_Kinji_general_2}
Let $E_{I^{p_1}_i}, E_{I^{p_2}_j}, \cdots, E_{I^{p_m}_k}
$ and $E_{I^{p_1}_i \sqcup \cdots  \sqcup I^{p_m}_k }$ be elements of basis of $ T_\mu$ .
Then the following is obtained
\begin{align}
[\phi_{\mu} ]( E_{I^{p_1}_i} ) \cdots [\phi_{\mu} ]( E_{I^{p_m}_k}  ) 
= [\phi_{\mu} ](E_{I^{p_1}_i}  \cdots  E_{I^{p_m}_k} ) +  \tilde{O}(\hbar({\mu}) ) .
\end{align}
\end{proposition}

\section{Examples}\label{sect5}
Let us see the examples of weak matrix regularization constructed in Section \ref{sect4}.
As the Lie algebra, we consider a semisimple Lie algebra. 
It corresponds to a classical solution of the 
mass-deformed IKKT matrix model.
We suppose a Lie-Poisson algebra as a classical space 
when the Lie algebra is regarded as a quantized space (fuzzy space).
$\mathfrak{su}(n)$ is a typical example of a semisimple Lie algebra.
In this section, $\mathfrak{su}(2)$ and $\mathfrak{su}(3)$ are
examined.
The $\mathfrak{su}(2)$ case is a well-known example of the fuzzy sphere, 
whereas the other cases provide examples of matrix regularizations 
that have not been previously studied.
In this section, attention is restricted to the weak matrix regularization mapped to irreducible representations. 
Consequently, one should note that, when considering each individual matrix regularization, there is inevitably a kernel originating, 
for instance, from Casimirs not included in the analysis. This issue will be addressed in Section \ref{rev_sect6}.

\subsection{
$\mathfrak{su}(2)$ ; Fuzzy space } 
The  fuzzy sphere is considered in {\rm \cite{matrix1,fuzzy1}}. See {\rm \cite{matrix1,fuzzy1,fuzzyb,fuzzyc}} for details. 
In \cite{Rieffel:2021ykh}, more general and mathematically precise statements are given. 
The fuzzy ${\mathbb R}^3$ is discussed in \cite{Hammou:2001cc}.
In this subsection, we reconstruct the fuzzy sphere using 
the method for constructing the weak matrix regularization 
of a Lie-Poisson variety in Section \ref{sect4} of this paper. 
In other words, we confirm that the matrix regularization in this paper is 
a generalization of the method for constructing the fuzzy sphere. \\

Let us consider $\mathfrak{su}(2)$ as $\mathfrak{g} $. 
The enveloping algebra of $\mathfrak{su}(2)$, $ \mathcal{U}_{\mathfrak{su}(2)}[\hbar] $,
is an algebra of all polynomials in $X_1  , X_2 , X_3 $ with relations
$ X_i  X_j  - X_j  X_i  - i \hbar  \epsilon^{ijk} X_k ~(i,j,k \in \{1,2,3 \} ) $.
Let $x_a~(1\le a\le 3)$ be commutative variables. 
$(x_1, x_2, x_3)=(x,y,z)$ is identified with the coordinates of ${\mathbb R}^3$.
The Lie-Poisson structure is defined by
\begin{align}
\{x_a, x_b\}= i \epsilon^{abc}x_c. \label{eq.poi}
\end{align}
$A_{\mathfrak{su}(2)}$ is given by $\mathbb{C}[x]$ with this Poisson bracket.
For arbitrary $f\in A_{\mathfrak{su}(2)}$ is given as
\begin{align*}
f=f_0+f_a x_a+\frac{1}{2}f_{ab}x_a x_b+\cdots ,
\end{align*}
where $f_{a_1\cdots a_i}\in \mathbb{C}$ is completely symmetric with respect to $a_1\cdots a_i$. 
Let $V^\mu$ be a vector space ${\mathbb{C}}^k$.
For example, we consider $V^2 = {\mathbb{C}}^2$, then the $q_{2}$ is given 
by a map from $A_{\mathfrak{su}(2)}$ to a matrix algebra ${\rm Mat}_{2}(\mathbb{C})$ is defined by
\begin{align*}
q_2 (f)&:=f_0{\bf 1}_2+f_a q_2 (x_a),\quad \quad q_2 (x_a):=\frac{\hbar}{2}\sigma^a ,
\end{align*}
where $\sigma^a$ is a Pauli matrix and ${\bf 1}_k$ is a $k\times k$ unit matrix. \\

In the case of $\dim V^\mu \ge 2$, the matrix regularization 
$q_k :A_{\mathfrak{su}(2)} \to  {\rm Mat}_{k}(\mathbb{C})$ is defined by
\begin{align*}
q_k(f)&:=f_0{\bf 1}_k+f_{a_1}q_k(x_{a_1})+\cdots +
\frac{1}{k!} 
f_{a_1\cdots a_{k}} q_k (x_{a_1}\cdots x_{a_{k-1}})\\
q_k (x_{a_1}\cdots x_{a_{m}})&:= \frac{\hbar^m(k)}{m!} \sum_{\sigma \in Sym (m) } J_{a_{\sigma (1)}} \cdots J_{a_{\sigma (m)}} 
\end{align*}
where $J_{a}$ are generators for the $k$-dimensional irreducible representation of $\mathfrak{su}(2)$ (the spin $2s+1=k$ representation). 
Each $q_k$ gives a map from a polynomial to a $k\times k$ matrix. 
$J_a $ satisfies
\begin{align}
[J_a , J_b]&= i \epsilon^{abc}J_c, \quad \quad [q_k (x_a), q_k (x_b)]=i \hbar^2(k) \epsilon^{abc}J_c=i\hbar(k) \epsilon^{abc} q_k (x_c). \label{eq.com}
\end{align}
Up to this point, we have not discussed the Casimir polynomial, so we have only been discussing the matrix regularization corresponding to the polynomial functions (Lie-Poisson algebra) defined 
on ${\mathbb R}^3$.
If the series $\hbar(k)$ converging to $0$ when $\dim V^\mu$ tends to infinity , the sequence of
$q_k$ is a matrix regularization of $A_{\mathfrak{su}(2)}$
whose corresponding Lie-Poisson albebra is defined on ${\mathbb R}^3$.
\\
\bigskip

From now on, we will consider the matrix regularization of the polynomial ring defined on the sphere by referring to the Casimir polynomial.
Solving $\{ x_a , f(x) \} = 0~  (a=1,2,3)$ for a 2nd-degree homogeneous polynomial $f \in A_{\mathfrak{su}(2)}$, we obtain 
a solution as a quadratic Casimir polynomial
\begin{align}\label{su(2)casimir}
f= \delta^{ab}x_a x_b .
\end{align}
Then $q_k (f(x)) =  \hbar^2(k) \delta^{ab}J_a J_{b}$ is a Casimir invariant:
\begin{align}
 \hbar^2(k)  \delta^{ab}J_a J_{a}= \hbar^2(k) \frac{1}{4}(k^2-1){\bf 1}_k  = : \lambda^2 {\bf 1}_k , \label{casimir_eigen_su(2)}
\end{align}
where the eigenvalue $ \lambda^2 $ is a non-negative constant.
We construct an ideal according to the method described in (\ref{4_22}) and the following.
So we choose the generators of the ideal $I(C) \subset A_{\mathfrak{su}(2)}$ as
\begin{align} \label{5_4}
f^C(x):= \delta^{ab}x_a x_b - \lambda^2  ,
\end{align}
and
\begin{align} \label{5_5}
I(C) := \left\{  f^C(x) g(x) \in  A_{\mathfrak{su}(2)} ~ | ~ 
g(x) \in  A_{\mathfrak{su}(2)} ~ \right\}  \subset A_{\mathfrak{su}(2)} .
\end{align}
Then, $A_{\mathfrak{su}(2)} /I(C)$ is a set of polynomials on $S^2$ given by
\begin{align}\label{sphere}
\delta^{ab}x_a x_b=\lambda^2.
\end{align}
From (\ref{casimir_eigen_su(2)}) and Section \ref{sect4_2}, we find that 
when $\hbar(k)$ is chosen as
\begin{align*}
\hbar(k) = \sqrt{\frac{4\lambda^2}{k^2-1}},
\end{align*} 
then we can define a weak matrix regularization of $A_{\mathfrak{su}(2)} /I(C)$.
To make a matrix regularization of $A_{\mathfrak{su}(2)} /I(C)$,
we also build the other parts.
\begin{align}
I(C(X))&:= 
 \left\{  \sum_{i,j} a_i (X) f^C(X) b_j (X) \in \mathcal{U}_{\mathfrak{su}(2)}[\hbar]~ | ~ 
a_i (X), b_j(X) \in  \mathcal{U}_{\mathfrak{su}(2)}[\hbar] ~ \right\} .
\end{align}
Here, $f^C(X)$ is a Casimir operator;
\begin{align}
f^C(X) : = q_U ( f^C(x) ) = \delta^{ab}X_a X_b - \lambda^2.
\end{align}
The reduced Gr\"obner basis for this $I(C)$ is given by 
$G= \{ \delta^{ab}x_a x_b - \lambda^2 \}$.
$ \mathcal{U}_{\mathfrak{su}(2)}[\hbar] / I(C(X)) ,  q_{U/I}$ and so on are defined by using these ideals and $G$.
A linear function
$\rho_{U/I , k} :  \mathcal{U}_{\mathfrak{su}(2)}[\hbar] / I(C(X))  \to gl (V^k )=gl ({\mathbb C}^k )$
is defined as follows. For any monomial $X^A = X_{a_1} \cdots X_{a_m} ~(m \in {\mathbb{N}})$
in $\mathcal{U}_{\mathfrak{su}(2)}[\hbar]$,
\begin{align}
\rho_{U/I , k} ( [X^A ] ) = \hbar^m (k) J_{a_1} \cdots J_{a_m}.
\end{align}
Finally the matrix regularization of $A_{\mathfrak{su}(2)} /I(C)$, is given as follows.
For $f= h + r_{f,G} ~( h \in I(C) ,  r_{f,G} \notin I(C))$ with
$$ \displaystyle
r_{f,G}= \sum_{A} c_{A} x^{A}
= \sum_{\deg x^{A} \le n_k } \!\!\!\!\!
c_{A} x^{A}
+
 \sum_{\deg x^{A} > n_k  } \!\!\!\!\!
c_{A} x^{A} ,
$$ the explicit calculation of 
$q_{A/I, k}  : A_{\mathfrak{su}(2)} / I(C)\to gl(V^k)$ is given as
\begin{align}
q_{A/I, k}  ([f(x)] ) 
= \sum_{m \le n_k } \!
c_{A}~ 
\frac{\hbar^m(\mu)}{m!} \sum_{\sigma \in Sym (m) } J_{a_{\sigma (1)}} \cdots J_{a_{\sigma (m)}} .
\end{align}
Thus, the matrix regularization of $A_{\mathfrak{su}(2)} / I(C)$, 
the fuzzy sphere, could be reconstructed by the method of this paper.\\

As an example, let us consider $f(x) = x^3 \in A_{\mathfrak{su}(2)}$.
By the graded lexicographic ordering,
$r_{f,G} = x(-y^2 -z^2 +\lambda^2)$.
Then 
\begin{align*}
q_{A/I, k}([x^3])= R_k \left( 
-\frac{\hbar^3(k)}{3} (J_1J_2^2 + J_2 J_1 J_2 +J_2^2 J_1) 
-\frac{\hbar^3(k)}{3} (J_1J_3^2 + J_3 J_1 J_3 +J_3^2 J_1) 
+\hbar(k) \lambda^2 J_1 
\right).
\end{align*}
If $n_k > 3~ (k >3)$, then  
$R_k$ is not different from a unit matrix 
in the above equation, i.e., $R_k (\cdots  ) = (\cdots  )$, then the result is 
\begin{align*}
q_{A/I, k}([x^3])= 
\hbar^3(k) J_1^3 +  \frac{ \hbar^3(k) }{3} J_1 .
\end{align*}

For example, we consider $q_3 : A_{\mathfrak{su}(2)}\to {\rm Mat}_{3}(\mathbb{C})$, generators $J_a$ for the $3$-dimensional irreducible representation of $\mathfrak{su}(2)$ given by
\begin{align*}
J_1 &= \begin{pmatrix}
0 & 0 & 0 \\
0 & 0 & i \\
0 & -i & 0
\end{pmatrix}, \quad
J_2 = \begin{pmatrix}
0 & 0 & -i \\
0 & 0 & 0 \\
i & 0 & 0
\end{pmatrix}, \quad
J_3 = \begin{pmatrix}
0 & i & 0 \\
-i & 0 & 0 \\
0 & 0 & 0
\end{pmatrix}.
\end{align*}
We construct the basis $E_i$ of $ {\rm Mat}_{3}(\mathbb{C})$ according to the construction methods $1$ to $5$
in Subsection \ref{3_2}. 
First, the generators $E_i:=\hbar J_i~(i=1,2,3)$ and the unit matrix $E_0:=Id_3$ are chosen as basis elements. 
Next, we construct an independent element 
from the symmetrized product of $E_i$. That is,
\begin{align*}
E_4:=&E_1^2=\frac{\hbar^2}{2}
\begin{pmatrix}
0&0&0\\
0&1&0\\
0&0&1
\end{pmatrix},\quad
E_5:=E_2^2=\frac{\hbar^2}{2}
\begin{pmatrix}
1&0&0\\
0&0&0\\
0&0&1
\end{pmatrix},\\
E_6:=&\frac{1}{2}(E_1E_2+E_2E_1)=\frac{\hbar^2}{2}
\begin{pmatrix}
0&-1&0\\
-1&0&0\\
0&0&0
\end{pmatrix},\\
E_7:=&\frac{1}{2}(E_2E_3+E_3E_2)=\frac{\hbar^2}{2}
\begin{pmatrix}
0&0&0\\
0&0&-1\\
0&-1&0
\end{pmatrix},\\
E_8:=&\frac{1}{2}(E_1E_3+E_3E_1)=\frac{\hbar^2}{2}
\begin{pmatrix}
0&0&-1\\
0&0&0\\
-1&0&0
\end{pmatrix}.
\end{align*} 
Since these $E_i~(i=0,\cdots ,8)$ are independent of each other, we obtained a basis. 
Consider $f=x^3+xy+x$ in $A_{\mathfrak{su}(2)}$. If $x>y>z$ with the graded lexicographic order as an ordering relation, then $r_{f,G}=(1+\lambda^2) x+xy-xy^2-xz^2$ since the reduced Gr\"obner basis 
$G= \{ z^2+y^2+x^2-\lambda^2 \}$ from 
$f^C(x,y,z)=z^2+y^2+x^2-\lambda^2$ of (\ref{4_22}) for $\mathfrak{su}(2)$. 
Note that $n_3 =2$ and $\lambda^2 = 2 \hbar^2(3)$,
\begin{align*}
q_{A\slash I, 3}([f])= \frac{\hbar(3)}{2} 
\begin{pmatrix}
0&-\hbar(3)&0\\
-\hbar(3)&0&-2i\\
0&2i&0
\end{pmatrix}
.
\end{align*}
For another case, let $g=z^3+z$. Then $g=r_{g,G}=z^3+z$, and
\begin{align*}
q_{A\slash I, 3}([g])= \hbar(3)  J_3 = E_3.
\end{align*}
The commutator of these is obtained by
\begin{align*}
[q_{A\slash I, 3}([f]),q_{A\slash I,3}([g])]&=\hbar^2(3) 
\begin{pmatrix}
-i\hbar(3) & 0 & 1\\
0 & i\hbar(3) & 0\\
-1 & 0 & 0
\end{pmatrix}.
\end{align*}
On the other hand, from
$r_{\{f,g\},G}= -i\left(\left(3 z^2+1\right) \left(\lambda^2 y+2 x^2 y-x^2-y^3+y^2-y z^2+y\right)\right)$,
we obtain
\begin{align*}
\hbar(3) q_{A\slash I, 3}(\{[f],[g]\})&= \hbar^2(3)
\begin{pmatrix}
-3i\hbar(3)&0&1\\
0&-i\hbar(3)&0\\
-1&0&-2i\hbar(3)
\end{pmatrix}.
\end{align*}
 The difference between them is 
\begin{align*}
[q_{A\slash I, 3}([f]),q_{A\slash I, 3}([g])]
- \hbar(3) q_{A\slash I, 3}(\{[f],[g]\}) = 2i\hbar^3(3)E_0 .
\end{align*}



\subsection{ $\mathfrak{su}(3)$; Fuzzy space}
Let us consider $\mathfrak{su}(3)$ as $\mathfrak{g}$. 
For a typical example of representation of $\mathfrak{su}(3)$, 
generators $T_i$ of $\mathfrak{su}(3)$ Lie algebra are given by
\begin{align*}
T_i=\frac{1}{2}\lambda_i ,
\end{align*}
where $\lambda_i$ are Gell-Mann matrices
\begin{align*}
\lambda_1 &= \begin{pmatrix}
0 & 1 & 0 \\
1 & 0 & 0 \\
0 & 0 & 0
\end{pmatrix}, \quad
\lambda_2 = \begin{pmatrix}
0 & -i & 0 \\
i & 0 & 0 \\
0 & 0 & 0
\end{pmatrix}, \quad
\lambda_3 = \begin{pmatrix}
1 & 0 & 0 \\
0 & -1 & 0 \\
0 & 0 & 0
\end{pmatrix},\\
\lambda_4 &= \begin{pmatrix}
0 & 0 & 1 \\
0 & 0 & 0 \\
1 & 0 & 0
\end{pmatrix}, \quad
\lambda_5 = \begin{pmatrix}
0 & 0 & -i \\
0 & 0 & 0 \\
i & 0 & 0
\end{pmatrix}, \quad
\lambda_6 = \begin{pmatrix}
0 & 0 & 0 \\
0 & 0 & 1 \\
0 & 1 & 0
\end{pmatrix},\\
\lambda_7 &= \begin{pmatrix}
0 & 0 & 0 \\
0 & 0 & -i \\
0 & i & 0
\end{pmatrix}, \quad
\lambda_8 = \frac{1}{\sqrt{3}} \begin{pmatrix}
1 & 0 & 0 \\
0 & 1 & 0 \\
0 & 0 & -2
\end{pmatrix}.
\end{align*}
From this basis, the structure constants are determined by
\begin{align}
f_{ab}^c = 2 ~ \mathrm{tr} [ T_a , T_b ] T_c .
\end{align}
In the following, this structure constant is fixed.
Using this structure constant, the Lie-Poisson structure is given as
\begin{align*}
\{x_a,x_b\}= f^{c}_{ab}x_c
\end{align*}
for $x_a (a= 1,2, \cdots , 8)$.
Solving $\{x_a,f(x)\}=0~(a=1,\cdots,8)$ for a $2$nd-degree homogeneous polynomial $f\in A_{\mathfrak{su}(3)}$, we obtain a quadratic Casimir polynomial 
\begin{align}
C^2 (x)=\frac{1}{3}\delta^{ab}x_ax_b. \label{su(3)_C2}
\end{align}
For the other case, a cubic Casimir Polynomial is given as
\begin{align*}
C^3 (x)&=\frac{1}{18} (2 \sqrt{3} x_8^3-6 \sqrt{3} x_1^2 x_8-6 \sqrt{3} x_2^2 x_8-6
   \sqrt{3} x_3^2 x_8+3 \sqrt{3} x_4^2 x_8+3 \sqrt{3} x_5^2 x_8+3 \sqrt{3}
   x_6^2 x_8+3 \sqrt{3} x_7^2 x_8\\
&\quad \quad -18 x_2 x_5 x_6+18 x_2 x_4 x_7-18 x_1(x_4 x_6+x_5 x_7)-9 x_3 \left(x_4^2+x_5^2-x_6^2-x_7^2\right) ).
\end{align*}
\begin{figure}[H]
  \begin{center}
    \includegraphics[width=70mm,height=70mm]{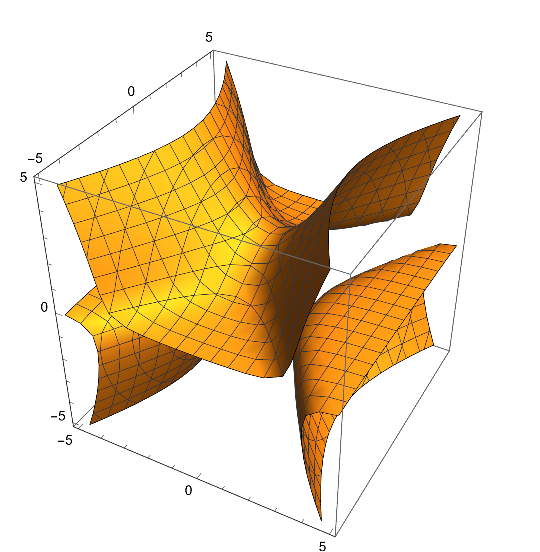}
  \end{center}
  \caption{Variety with $C^3 (x)=1$. ($x_2=x_3=x_4=x_5=x_7=0$) }
  \label{fig:c3variety.eps}
\end{figure}
From the above considerations, the spaces in which a Lie-Poisson algebra 
$A_{\mathfrak{su}(3)}/ I(C)$ can be defined are
${\mathbb R}^8$, $C^2 (x)= const$ in ${\mathbb R}^8$, i.e. $S^7$, 
and $C^3 (x) = const$ in ${\mathbb R}^8$.\\
\bigskip

i) Case of ${\mathbb R}^8$.\\
When the series $\hbar(\mu)$ converging to $0$ in the limit as $\dim V^\mu$ approaches infinity,
and the eigenvalues of the Casimir operators are not fixed,
$q_\mu$ is a matrix regularization of $A_{\mathfrak{su}(3)}$.
\\

ii) Case of $S^7$ : (fixing each $\hbar(\mu)$ by $q_\mu (C^2(x) - \lambda^2) = 0$). \\
Consider $\dim V^3 = 3 $ case as a simple example.
(We denote $\mu$ of this example by $3$.)
$E_i=\hbar(3) T_i~(i=1,\cdots, 8)$ and $E_0=Id_3$ yield a basis of ${\rm Mat}_3( \mathbb{C} )$, 
so $n_3=1$. 
For example, $f=x_1 +x_1^2$ in $A_{\mathfrak{su}(3)}$. 
In this case, the Gr\"obner basis is given by $\{ C^2 (x)- \lambda^2 \}$, then
\begin{align*}
r_{f,G}= x_1 + 3 \lambda^2-x_2^2-x_3^2-x_4^2-x_5^2-x_6^2-x_7^2-x_8^2.
\end{align*}
Since $n_3 = 1$, we simply 
replace each $x_i$ with $\hbar(3) T_i$, then
\begin{align*}
q_{A\slash I, 3}([f])= E_1.
\end{align*}
Suppose $g(x)= x_2 \in  A_{\mathfrak{su}(3)}$.
Since $g=r_{g,G}=x_2$, 
the commutator of them is given by
\begin{align*}
[q_{A\slash I, 3}([f]),q_{A\slash I, 3}([g])]=& i\hbar(3) E_3.
\end{align*}
From the Lie-Poisson $\{[f],[g]\} = [i x_3+  2i x_1 x_3 ]$, the following is obtained:
\begin{align*}
\hbar(3) q_{A\slash I, 3}(\{[f],[g]\})=&
\hbar(3) q_{A\slash I, 3}([i x_3+  2i x_1 x_3 ])= i \hbar(3)E_3
\end{align*}
Therefore, in this case,
$[q_{A\slash I, 3}([f]),q_{A\slash I, 3}([g])] = 
\hbar q_{A\slash I, 3}(\{[f],[g]\})$.
This result is expected from (\ref{prop_deg1}).

\bigskip

iii) Case of $C^3(x) =\lambda^3$ in ${\mathbb R}^8$:
(fixing $\hbar(\mu) $ by $q_\mu (C^3 (x) - \lambda^3)=0$). \\
Consider the same case $\dim V^\mu = 3 $  as $S^7$ case. 
To distinguish it from Case ii) (case of $S^7$), we do not assign a specific number to 
$\mu$, and instead keep it as 
$\mu$ in the following discussion.
Note also that the values of $\hbar(3)$ in Case ii) and 
$\hbar(\mu)$ in the following are different.
The basis is the same as in Case ii), i.e., 
$E_i=\hbar(\mu) T_i~(i=1,\cdots, 8)$ and $E_0=Id_3$ yield the basis. $n_\mu=1$ is also the same as in Case ii). 
The Gr\"obner basis is given by $\{ C^3 (x) - \lambda^3 \}$.
As an example, let us consider 
$f_2=x_1^2x_3x_8+x_2$.
Then $r_{f_2 , G}$ is given by
\begin{align*}
&x_2-\sqrt{3} \lambda^3 x_3-\sqrt{3} x_1 x_3 x_4
   x_6-\sqrt{3} x_1 x_3 x_5 x_7-x_2^2 x_3
   x_8+\sqrt{3} x_2 x_3 x_4 x_7-\sqrt{3} x_2
   x_3 x_5 x_6-x_3^3  x_8
   \\
& +\frac{1}{2}\left(-
   \sqrt{3} x_3^2 x_4^2- \sqrt{3} x_3^2
   x_5^2+\sqrt{3} x_3^2 x_6^2+ \sqrt{3}
   x_3^2 x_7^2+ x_3 x_4^2 x_8+
   x_3 x_5^2 x_8+ x_3 x_6^2
   x_8+ x_3 x_7^2 x_8 \right)+\frac{x_3
   x_8^3}{3}.
\end{align*}
Since $\lambda^3$ is in proportional to $\hbar^3(\mu)$, we get
\begin{align*}
q_{A\slash I,\mu}([f_2]) = \hbar(\mu) T_2  = E_2 . 
\end{align*}
Let us consider $g_2 = r_{g,G}=x_4$.
The similar calculations as in the case of $S^7$ yield
\begin{align*}
[q_{A\slash I,\mu}([f_2]),q_{A\slash I,\mu}([g_2])]& = \frac{i \hbar(\mu)}{2} E_6.
\end{align*}
The matrix regularization for Poisson brackets can also be performed straightforwardly, 
yielding the following, 
\begin{align*}
\hbar(\mu)  q_{A\slash I,\mu}(\{[f_2],[g_2]\})&=[q_{A\slash I,\mu}([f_2]),q_{A\slash I,\mu}([g_2])].
\end{align*}


\section{Reducible representations, coadjoint orbits, and discussions}\label{rev_sect6}

Thus far, our discussion has focused solely on the quantization sequence of irreducible representations.
Focusing on a single quantization corresponds to a single irreducible representation, and its geometric counterpart is the coadjoint orbit \cite{{Kirillov2004},{Kostant1970},{Woodhouse1991},{Souriau1997}}.
In other words, once a representation is fixed by its Dynkin label or highest weight, the corresponding coadjoint orbit is determined.
The coadjoint orbit is identified by the eigenvalues of all Casimir operators.
Accordingly, when the discussion is restricted to irreducible representations, a situation may arise in which a single quantization is associated, but its kernel becomes large and corresponds to a proper subspace of the space under consideration.
For example, in the previous section we discussed the quantization of $\mathfrak{su}(3)$ case. 
In that context, what was described as the quantization of ${\mathbb{R}}^8$
 or $S^7$ actually corresponds only to a subspace when the discussion is restricted to irreducible representations, despite the fact that quantizations of them are associated.
To clarify this point further by means of the $\mathfrak{su}(3)$
example: when the quotient space is constructed as a Lie-Poisson algebra using the quadratic Casimir, the domain of definition of the quantization becomes the algebra of functions on  $S^7$.
However, once a particular representation $\rho_{U\slash I,\mu}$
 is chosen, this representation necessarily admits the cubic Casimir, and consequently elements corresponding to the relations generated by the cubic Casimir may lie in the kernel of the quantization 
 of $A_\mathfrak{g}  /I(C)$.
\\

Since we are, from the outset, considering quantizations mapping into finite-dimensional matrices, the existence of a kernel in the quantization, or a loss of faithfulness, must be regarded as assumed a priori.
Nevertheless, with respect to kernels arising from the intrinsic properties of the classical space, it is natural to regard their elimination as providing a more faithful reflection of the underlying classical geometry.
In this section, on the basis of this perspective, we extend the discussion to the case of reducible representations of quantization.
In so doing, we examine the possibility that various coadjoint orbits may be incorporated so as to encompass a substantial portion of the classical space.

\subsection{Generalization to reducible representations}

\label{rev_sect6.1}
We consider the following decomposition of a reducible representation into a direct sum of irreducible representations of a semi-simple Lie algebra $\mathfrak{g}$
\begin{align*}
V^\mu &=\bigoplus_{a=1}^{m_\mu} V_\mu ^a.
\end{align*}
We define a projection $\hat{P}_a:V^\mu\to V_\mu^a$, and a reducible representation $\rho^\mu$ of $\mathfrak{g}$
\begin{align*}
\rho^\mu &= \sum_a^{m_\mu} \rho_{\mu ,a}\hat{P}_a,\\
\rho_{\mu,a}&:\mathfrak{g}\to End(V_\mu ^a).
\end{align*}
Suppose that each representation satisfies $[\rho_{\mu,a}(e_i),\rho_{\mu,a}(e_j)]=f^k_{ij}\rho_{\mu,a}(e_k)$. Then
\begin{align*}
[\rho ^\mu(e_i),\rho^\mu(e_j)]&=\sum_{a,b}[\rho_{\mu,a}(e_i)\hat{P}_a,\rho_{\mu,b}(e_j)\hat{P}_b]
=f^k_{ij}\rho^\mu(e_k).
\end{align*}
Let us assume that $\hbar(\mu)$ is given for each representation
\begin{align*}
\hat{\hbar}(\mu) = \hbar_1(\mu )Id_1\oplus \cdots \oplus \hbar_{m_\mu}(\mu)Id_{m_\mu} = \hbar(\mu ) 
( r_1(\mu )Id_1\oplus \cdots \oplus r_{m_\mu}(\mu)Id_{m_\mu} ) 
,
\end{align*}
where $Id_a:= Id_{V_{\mu^a}}~ (a= 1, \dots , m_\mu)$, $\hbar_a(\mu) \in {\mathbb R}~ (a= 1, \dots , m_\mu)$, 
and $r_a(\mu) \in {\mathbb R}~ (a= 1, \dots , m_\mu)$
satisfies $\hbar(\mu) r_a(\mu) = \hbar_a(\mu) $. 
Here, the direct sum of algebras is represented 
by block-diagonal matrices. 
Following this convention, we shall also describe direct product of non-matrix algebras using the direct sum notation in what follows.
We also define each basis as follows
\begin{align*}
e^{(\mu)}_i&:=\hat{\hbar}(\mu)\rho^\mu(e_i)=\hbar_1(\mu)\rho_{\mu,1}(e_i)\oplus \cdots \oplus \hbar_{m_\mu}(\mu)\rho_{\mu,m_\mu}(e_i),\\
e^{(\mu,a)}_i&:=\hbar_a(\mu)\rho_{\mu,a}(e_i).
\end{align*}
From the commutation relations of the representation $\rho_{\mu,a}$,
\begin{align*}
[e^{(\mu,a)}_i,e^{(\mu,a)}_j]&=\hbar_a(\mu)f^k_{ij}e^{(\mu,a)}_k,\\
[e_i^{(\mu)},e_j^{(\mu)}]&=\hat{\hbar}(\mu)f^k_{ij}e^{(\mu)}_k.
\end{align*}
For $T_\mu^a:=\langle e^{(\mu,a)}\rangle$, we redefine $T_\mu:=\langle e^{(\mu)} \rangle=T^1_\mu\oplus \cdots \oplus T^{m_\mu}_\mu$ as an $R$-algebra generated by $e^{(\mu)}$. Given the basis $E_i^a$ $(i=1,\cdots,D_a)$ of $T_\mu^a$, we set $n_\mu^a:=\max(\deg E_1^a,\cdots, \deg E_{D_a}^a)$, and proceed to redefine the quantization map for the reducible representations.
\begin{definition}\label{def6_1}
We define a linear function $q_\mu^a:A_{\mathfrak{g}}\to T_\mu^a$ by
\begin{align*}
q^a_\mu \Big(\sum_I f_Ix^I \Big)=\sum_{|I|\le n^a_\mu}f_I e^{(\mu,a)}_{(I)},
\end{align*}
and a linear function $q^R_\mu:A_\mathfrak{g}\to T_\mu$ by
\begin{align*}
q^R_\mu \Big( \sum_I f_Ix^I \Big)=\sum_{a=1}^{m_\mu}\sum_{|I|\le n_\mu^a}f_Ie_{(I)}^{(\mu,a)}\hat{P}_a=\sum_{a=1}^{m_\mu} q^a_\mu \Big(\sum_If_Ix^I \Big)\hat{P}_a,
\end{align*}
i.e. $q^R_\mu:=q^1_\mu\oplus \cdots \oplus q^{m_\mu}_\mu$.
\end{definition}
Under these settings, results analogous to similar lemmas, propositions and theorems stated in Section \ref{sect3} can be derived by direct calculation, as follows.
\begin{theorem}
Let $\{ q_\mu^R:A_\mathfrak{g}\to T_\mu \}$ be a sequence of quantizations defined by Definition \ref{def6_1}. 
Suppose that $\{\hbar(\mu)\}$ is a sequence such that $\hbar_a(\mu)\to 0 ~(a= 1,2, \cdots , m_\mu )$ as $\dim V_{\mu} \to \infty$. Then $q_\mu^R$ is a weak matrix regularization. In other words,
\begin{align*}
[q^R_\mu(f),q^R_\mu(g)]=\hat{\hbar}(\mu)q^R_\mu(\{f,g\})+\hat{\hbar}^2(\mu)P,
\end{align*}
where $P=\sum_{a=1}^{m_\mu}\sum_la_{i_1\cdots i_l}^{\mu,a}q^a_\mu(x_{i_1})\cdots q^a_{\mu}(x_{i_l})\hat{P}_a$.
\end{theorem}
\begin{proposition}
For any $f,g\in A_{\mathfrak{g}}$ with $\deg f+\deg g \le \min(n_\mu^1,\cdots,n_\mu^{m_\mu})$, there exists 
$\displaystyle P=\sum_{a=1}^{m_\mu}\sum_{i=1}^{D_a} c_i^a(\hbar_a(\mu))E^a_i\hat{P}_a \in T_\mu$ with $c_i^a(\hbar_a(\mu))\in \mathbb{C}[\hbar_a(\mu)]$ satisfying
\begin{align*}
q_\mu^R(f)q^R_\mu(g)=q^R_\mu(fg)+\hat{\hbar}(\mu)P.
\end{align*}
\end{proposition}
\bigskip
We consider a $k$th-degree Casimir operator $C_{k,i}^{(\mu,a)}=\Lambda_i^{a,k}(V_\mu^a)Id_{\mu,a}$ of each $V_\mu^a$ such that $|\Lambda_i^{a,k}(V_\mu^a)|\to \infty $ when $\dim V_\mu^a\to \infty$. For a fixed Casimir polynomial $C_i^k(x)=\sum_J C^k_{i,J}x^J$ in $A_\mathfrak{g}$, we defined $C^{k,a}_i(e^{(\mu,a)})$ by
\begin{align*}
C^{k,a}_i(e^{(\mu,a)})&:=q^a_\mu(C^k_i(x))\\
&=\frac{(\hbar_a(\mu))^k}{k!}\sum _J C_{i,J}^k \sum_{\sigma \in S_k}\rho_{\mu,a}(e_{i_\sigma(1)})\cdots \rho_{\mu,a}(e_{i_\sigma (k)})\\
&=(\hbar_a(\mu))^k \Lambda_i^{a,k}(V_\mu^a)Id_{\mu,a}.
\end{align*}

As in the case of irreducible representations, we fix the parameter 
$\lambda_i^k\in \mathbb{C}$ determining the ideal, and choose each $\hbar_a(\mu) $  so as to satisfy the following condition
\begin{align*}
C_i^{k,a}(e^{(\mu,a)})=(\hbar_a(\mu))^k C_{k,i}^{(\mu,a)}=\lambda^k_i,
\end{align*}
for any $q_\mu^a:A_\mathfrak{g}\to V^\mu_a (a= 1, \dots , m_\mu )$. 
$\lambda_i^k\in \mathbb{C}$ is the eigenvalue of the matrix $C_i^{k,a}(e^{(\mu,a)})$.
Then, we obtain
\begin{align*}
q^R_\mu(C^k_i(x))&=\sum_{a=1}^{m_\mu} q^a_\mu(C^k_i(x))\hat{P}_a
=\lambda^k_i Id_{V_\mu}.
\end{align*}
We also introduce the direct product algebra of the enveloping algebra with $\hbar$, and define the quantization as follows.
\begin{definition}
For a set of parameters $\overrightarrow{\hbar}_\mu=(\hbar_1,\cdots,\hbar_{m_\mu})$, we define a canonical linear map $q_U^{\overrightarrow{\hbar}}:A_\mathfrak{g}\to \mathcal{U}_\mathfrak{g}[\hbar_1]\oplus \cdots \oplus \mathcal{U}_\mathfrak{g}[\hbar_{m_\mu}]$ by
\begin{align*}
q^{\overrightarrow{\hbar}}_U:=\left . q_U\right |_{\hbar=\hbar_1}\oplus \cdots \oplus \left .q_U\right |_{\hbar=\hbar_{m_\mu}}.
\end{align*}
\end{definition}
Under this settings, similarly to Section \ref{sect4}, some lemmas, propositions, and theorems analogous to those in Section \ref{sect4} are derived by direct calculation.
Since the proofs consist of calculations almost parallel to those above, we omit all the details here.
\begin{proposition}
$q_U^{\overrightarrow{\hbar}}(f): A_\mathfrak{g}\to \mathcal{U}[\hbar_1]\oplus \cdots \oplus \mathcal{U}[\hbar_{m_\mu}]$ is a quantization i.e.,
\begin{align*}
[q^{\overrightarrow{\hbar}}_U(f),q^{\overrightarrow{\hbar}}_U(g)]=\bigoplus_{i=1}^{m_\mu}(\hbar_iq_U(\{f,g\})+\tilde{O}(\hbar_i^2)).
\end{align*}
for $f,g\in A_{\mathfrak{g}}$. Especially, if $\min\{\deg f,\deg g\}\le 1$, then
\begin{align*}
[q^{\overrightarrow{\hbar}}_U(f),q^{\overrightarrow{\hbar}}_U(g)]=\bigoplus_{i=1}^{m_\mu}(\hbar_iq_U(\{f,g\})).
\end{align*}
\end{proposition}
\begin{definition}
We consider the following as the restriction of $I(C(X))$ to $\hbar = \hbar_i$
\begin{align*}
\left .I(C(X))\right|_{\hbar =\hbar_i}=\Big\{\sum_{i,j,k}a_j(X)f_i^C(X)\left.b_k(X)\right|_{\hbar=\hbar_i}\mid f^C_i(X):=q_U(f^C_i(x))\Big\}.
\end{align*}
We defined a linear map $q^R_{U\slash I}$ by
\begin{align*}
\begin{array}{rccc}
q^R_{U\slash I}:&A_\mathfrak{g}\slash I(C)&\to &\displaystyle \bigoplus _{i=1}^{m_\mu}(\mathcal{U}_\mathfrak{g}[\hbar_i]\slash \left.I(C(X))\right|_{\hbar=\hbar_i})\\
{}&\rotatebox{90}{$\in$}&{}&\rotatebox{90}{$\in$}\\
{}&[f]=[r_f]&\mapsto &\displaystyle q^R_{U\slash I}([f])=\bigoplus_{i=1}^{m_\mu} \left.[q_U(r_f)]\right|_{\hbar=\hbar_i}.
\end{array}
\end{align*}
\end{definition}
Then the following is obtained.
\begin{theorem}
The above $q^R_{U\slash I}$ is a quantization, i.e.
\begin{align*}
[q^R_{U\slash I}([f]),q^R_{U\slash I}([g])]=\hat{\hbar}q^R_{U\slash I}(\{[f],[g]\})+\tilde{O}(\hat{\hbar}^2),
\end{align*}
where $\hat{\hbar}=\hbar_1 \oplus \cdots \oplus \hbar_{m_\mu}$,
and $\tilde{O}(\hat{\hbar}^2) =  \tilde{O}({\hbar_1}^2)\oplus \cdots \oplus \tilde{O}({\hbar_{m_\mu}}^2)$.
\end{theorem}

For each $\mathcal{U}_\mathfrak{g}[\hbar_a]$, 
after identifying $\hbar_a$ with 
$\hbar_a (\mu)$, we attach the following representation.
\begin{align*}
\begin{array}{lccc}
\rho_{U,\mu}^a:&\mathcal{U}_{\mathfrak{g}}[\hbar_a]&\to &gl(V_\mu^a)\\
{}&\rotatebox{90}{$\in$}&{}&\rotatebox{90}{$\in$}\\
{}&X_i&\mapsto&\rho_{U,\mu}^a(X_i)=e_i^{(\mu,a)}.
\end{array}
\end{align*}
Using these $\rho_{U,\mu}^a$, we introduce the representation
\begin{align*}
\begin{array}{lccc}
\rho_{U,\mu}^R:&\bigoplus_{i=1}^{m_\mu} \mathcal{U}_{\mathfrak{g}}[\hbar_i]&\to &gl(V_\mu^1)\oplus \cdots \oplus gl(V_\mu^{m_\mu})\subset gl(V_\mu)
\end{array}
\end{align*}
by $\rho_{U,\mu}^R :=\bigoplus_{a=1}^{m_\mu}\rho^a_{U,\mu}$.
As in Subsection \ref{sect4_2},
we chose $I(C)$ to be the one given in  (\ref{4_23}).
It is also clear that the following definition is well-defined.
\begin{align*}
\begin{array}{lccc}
\rho_{U\slash I,\mu}^a:&\mathcal{U}_\mathfrak{g}[\hbar_a]\slash \left.I(C(X))\right|_{\hbar=\hbar_a}&\to &gl(V_\mu^a)\\
{}&\rotatebox{90}{$\in$}&{}&\rotatebox{90}{$\in$}\\
{}&[X^I]_a&\mapsto&\rho_{U\slash I,\mu}^a([X^I]_a)=
\rho_{U,\mu}^a(X^I) =
e^{(\mu,a)}_{(I)}.
\end{array}
\end{align*}
The representation
$\rho_{U\slash I,\mu}^R: \bigoplus _{i=1}^{m_\mu}\mathcal{U}_\mathfrak{g}[\hbar_i]\slash \left.I(C(X))\right|_{\hbar=\hbar_i} \to \bigoplus_{i=1}^{m_\mu} gl(V_\mu^i)\subset gl(V^\mu)$
is defined by
$\rho^R_{U\slash I,\mu}=\bigoplus_{i=1}^{m_\mu} \rho_{U\slash I,\mu}^i.$
For an algebra homomorphism $\rho^R_{U\slash I,\mu}=\bigoplus_{i=1}^{m_\mu} \rho_{U\slash I,\mu}^i$, 
If $\forall i$, $[f^i]_i\in \mathcal{U}_\mathfrak{g}[\hbar_i]\slash \left.I(C(X))\right|_{\hbar=\hbar_i}$ belongs to $\tilde{O}((\hat{\hbar}_i(\mu))^n)$ then
\begin{align*}
\rho^R_{U\slash I,\mu}([f^1]_1\oplus \cdots \oplus [f^m]_{m_\mu})=\tilde{O}((\hat{\hbar}(\mu))^n):=
\tilde{O}((\hbar_1(\mu))^n)\oplus \cdots \oplus 
\tilde{O}((\hbar_{m_\mu}(\mu))^n).
\end{align*}
\begin{theorem}
For $q^{R,pre}_\mu:= \rho_{U\slash I,\mu}^R\circ q^R_{U \slash I}$,
\begin{align*}
[q^{R,pre}_\mu([f]).q^{R,pre}_\mu([g])]=\hat{\hbar}(\mu)q^{R,pre}_\mu(\{f,g\})+(\hat{\hbar}(\mu))^2P_\mu,
\end{align*}
where $\displaystyle P_\mu =  \sum_a^{m_\mu}\sum_i^{D^a}
c_i^a (\hbar(\mu)) E_i^a \hat{P}_a \in T_\mu$
with $c_i^a(\hbar_a(\mu))\in \mathbb{C}[\hbar_a(\mu)]$.
\end{theorem}
The projection map $\displaystyle R^a_\mu(M(\hbar_a(\mu)))=\sum_{0\le k\le n^a_\mu}\hbar^k_aM_k$
is also introduced similarly for any $a=1,\cdots , m_\mu$,
where $\displaystyle M(\hbar_a(\mu))=\sum_{0\le k}(\hbar_a(\mu))^k M_k$~ $(M_k\in gl(V_\mu^a))$.
Using them,
$\displaystyle R_\mu^R:=\bigoplus _{i=1}^{m_\mu} R^i_\mu$
is defined.
Finally, we define the quantization as follows.
\begin{definition}
A quantization of $A_\mathfrak{g}  /I(C) \to  gl (V^\mu )$ is defined as
\begin{align}
q_{A/I ,\mu}^R : = R_\mu^R \circ q^{R, pre}_\mu = R_\mu^R \circ  \rho_{U/I , \mu}^R \circ q^R_{U \slash I}.
\end{align}
\end{definition}
Finally, we obtain the following.
\begin{theorem}
$q_{A/I, \mu}^R $
is a weak matrix regularization, i.e., for  any $[f], [g] \in A_\mathfrak{g} / I(C)$, 
there exists $\displaystyle P = \sum_{a=1}^{m_\mu} \sum_i^{D_a} c_i^a(\hbar_a(\mu)) E_i^a \hat{P}_a \in T_\mu$,
where each $c_i^a(\hbar_a(\mu)) $ is a polynomial in $\hbar_a(\mu)$, such that
\begin{align}
[ q_{A/I, \mu}^R ( [ f ] ) ,  q_{A/I, \mu}^R ( [ g ] ) ] = \hat{\hbar}(\mu ) q_{A/I, \mu}^R ( \{ [f] , [g] \} ) + (\hat{\hbar}(\mu))^2 P.
\label{q_AI_R_quantization}
\end{align}
Furthermore, 
if $\displaystyle \deg r_f +\deg r_g \le \min_{1\le a \le m_\mu} \{n_\mu^a\}$ and $\displaystyle \deg r_{fg} \le  \min_{1\le a \le m_\mu} \{n_\mu^a\}$,
$q_{A/I, \mu}^R $ also satisfies 
$$
q_{A/I, \mu}^R ( [ f ] )  q_{A/I, \mu}^R ( [ g ] ) = q_{A/I, \mu}^R ( [f] [g]) + \hbar(\mu)  P 
$$
with some $\displaystyle P = \sum_{a=1}^{m_\mu} \sum_i^{D_a} c_i^a(\hbar_a(\mu)) E_i^a \hat{P}_a  \in T_\mu$.
\end{theorem}

In this section, we have seen that the quantization can also be constructed for reducible representations.
We reconstruct a weak matrix regularization.
It is possible to choose the series $\{ V^\mu \}$ such that
the number of distinct types of representations increases without bound.
As an example, let us consider ${\mathfrak{su}(2)} $ case. 
Let $m_\mu$ be taken as a diverging increasing sequence, and choose the representation spaces 
$V_\mu^a ~(a=1,2, \dots , m_\mu)$ so that the radii of the spheres $\sqrt{\lambda^2_a}$
become dense in ${\mathbb R}_{>0}$ in the limit. 
In other words, it should be possible to take them in a manner analogous to a singular foliation, filling the interior of 
${\mathbb R}^3$ densely with spheres. 
Such a construction may be regarded as a good approximation of ${\mathbb R}^3$.
This idea is already discussed in \cite{Vitale:2012dz}.
In the case of $\mathfrak{su}(3)$, the situation becomes more complicated.   
We fix a representation by its highest weight and construct the coadjoint orbit by acting adjointly with 
$SU(3)$ on a traceless diagonal matrix $\zeta=\mathrm{diag}(\zeta_1, \zeta_2, \zeta_3)$.  
When the eigenvalues are non-degenerate, the coadjoint orbit corresponds to $SU(3)/T^2$, which has dimension $6$, 
whereas when two of the eigenvalues coincide, it corresponds to ${\mathbb C}P^2$, which has dimension $4$.  
These coadjoint orbits are submanifolds of $S^7$.  
As in the case of $\mathfrak{su}(2)$, it is possible to consider a sequence of reducible representations that embeds infinitely many copies of various 
$SU(3)/T^2$ and ${\mathbb C}P^2$ orbits into $S^7$.  
For example, each $SU(3)/T^2$ orbit can be embedded into $S^7$ as a closed submanifold, 
and the family of regular orbits densely fills $S^7$ in the following sense. 
Introduce the equation of $S^7$ by considering a traceless $3 \times 3$ Hermitian matrix $A$ 
with $\mathrm{tr}(AA^\dagger) = r^2$.  
Since it can be diagonalized by the adjoint action, the above $\zeta$ may be regarded as a representative element. 
Identifying this with the highest weight, one can view $(\zeta_1 - \zeta_2, \zeta_2 - \zeta_3)$ as the corresponding Dynkin labels.  
For an arbitrary point on $S^7$, the ratios are real numbers, but when sufficiently large representations are considered, 
they can be approximated by integer ratios, which implies that a corresponding representation exists.  
In this way, one may expect to construct a sequence of representations that fills $S^7$.  
When the eigenvalues of $\zeta$ degenerate, singularities appear, i.e. the corresponding orbit becomes ${\mathbb C}P^2$ or a point, and this requires caution.  \\

Furthermore, what happens in the case of a general Lie algebra?
To what extent this strategy is effective, and for which types of Lie-Poisson algebras, must be investigated for each individual Lie-Poisson algebra, 
and this will be left as a subject for future work.

\section{Summary}\label{sect6} 

In this paper, the quantization of the Lie-Poisson algebra
was carried out as a matrix regularization in a weak sense. \\

In Section \ref{sect2}, it was shown that the mass-deformed IKKT
matrix model is equivalent to the matrix model whose solution is
a basis of a semisimple Lie algebra.
From this fact, a basis of every semisimple Lie algebra makes 
a classical solution of the mass-deformed IKKT matrix model. 
The precise statement of this claim is given in Theorem \ref{thm_2_2}.
Lie-Poisson algebras are expected as commutative limits
of the algebras generated by these classical solutions.
Matrix regularization connects the Lie-Poisson algebra and the algebra generated by a Lie algebra as a quantization.
So, the matrix regularization of the Lie-Poisson algebras was
studied. 
It is a generalization of the method for constructing the fuzzy sphere.
For a long time, the enveloping algebras of Lie algebras have been studied 
as a certain quantization of Lie-Poisson algebras.
Giving a matrix representation of the enveloping algebra 
roughly corresponds to this matrix regularization.
In this paper, we constructed a quantization by relaxing the standard conditions of matrix regularization.
So we called it ``weak matrix regularization'' for the sake of distinction.
\\

The process of constructing the weak matrix regularization for 
$A_\mathfrak{g} / I(C)$
is a little complicated, so it would be better to review the procedure here,
where $\mathfrak{g}$ is a $d$-dimensional Lie algebra,  and $A_\mathfrak{g} / I(C)$ is a
Lie-Poisson algebra.
We assume that the Lie algebra $\mathfrak{g}$ together with a sequence of its representations satisfies the following conditions: 
its Casimir operators are proportional to the identity operator, and their eigenvalues diverge in the limit 
where the dimension of the representation space tends to infinity.
For example, semisimple Lie algebras are Lie algebras that satisfy this condition.
The ideal $I(C)$ is not arbitrary, but is assumed to be made from $k$th-degree 
Casimir polynomials.
These polynomials determine the geometry of $A_\mathfrak{g} / I(C)$.
The following is a summary of the procedure for matrix regularization of $A_\mathfrak{g} / I(C)$.
 \\
1). At first we construct a matrix regularization with the trivial ideal $I(C) = \{ 0 \}$, i.e.,
 $A_\mathfrak{g} / I(C) = A_\mathfrak{g} $.
 This is given in Subsection \ref{3_2}. \\
We consider a representation of $\mathfrak{g}$ to $T_\mu \subset gl (V^\mu )$
 which is the algebra generated by the image of the representation.
We set $n_\mu$ as a certain degree that determines the kernel of the quantization.
The linear function
$q_\mu : A_\mathfrak{g} \rightarrow T_\mu $
defined by
$
q_\mu ( \sum_k f_{i_1, \cdots , i_k} x^{i_1} \cdots x^{ i_k} ) 
=
\sum_k^{n_\mu} f_{i_1, \cdots , i_k} e^{(\mu)}_{(i_1, \cdots , i_k)} 
$
is the matrix regularization for $A_\mathfrak{g} $.
\begin{align*}
\vcenter{
\xymatrix@C=16pt@R=4pt{
A_\mathfrak{g}  \ar[r]^{q_{\mu} \in Q}  
& T_\mu \subset gl (V^\mu )
\\
{}&{}\\
{}\ar@{(-}[u] &{}\ar@{(-}[u] \\
f(x)=  \sum_k f_{i_1, \cdots , i_k} x^{i_1} \cdots x^{ i_k} \ar@{|->}[r]& 
\sum_k^{n_\mu} f_{i_1, \cdots , i_k} e^{(\mu)}_{(i_1, \cdots , i_k)} 
}}
\end{align*}
2). We introduce a quantization map $q_U$ from $A_\mathfrak{g} $ 
 to  enveloping algebra $ \mathcal{U}_\mathfrak{g} $
 by $q_U (x_{\alpha_1} \cdots x_{\alpha_m}) = X_{(\alpha_1, \cdots , \alpha_m )}$.
This quantization is basically well known from long ago.
\\
3). Next, we construct a quantization map 
$q_{U/I}  : A_\mathfrak{g}  /I(C) \to  \mathcal{U}_\mathfrak{g} [\hbar]/I(C(X))$
for nontrivial $I(C)$, which is described in Section \ref{4_1}. 
$I(C)$ is not arbitrary, but is assumed to be made from $k$th-degree 
Casimir polynomials satisfying (\ref{fixed_casimir_relation}).
We use $q_\mu$ in 1) to obtain the relation (\ref{fixed_casimir_relation}).
Let $G$ be the reduced Gr\"obner basis of $I(C)$.
For any $f(x) \in  {\mathbb C}[x]$ $f(x)= r_f (x) + h_f(x)$ is uniquely determined by $G$, 
where $h_f(x) \in I(C)$ and $r_f(x) \notin I(C)$.
Then we can define $q_{U/I}$ by $q_{U/I}  ([f(x)]) :=  [ q_{U}  ( r_{f, G}(x) ) ]$.\\
4). There exists an algebra homomorphism 
$\rho_{U/I , \mu} :  \mathcal{U}_\mathfrak{g}[\hbar] / I(C(X))  \to gl (V^\mu )$.
Using this $\rho_{U/I , \mu}$ and a projection operator
$R_\mu : gl (V^\mu )[\hbar(\mu)]  \to gl (V^\mu ) [\hbar(\mu)] $ 
that restricts the degree of $\hbar$ to $n_\mu$ or less. 
Finally, we get the weak matrix regularization $A_\mathfrak{g}  /I(C) \to  gl (V^\mu )$ by
$q_{A/I ,\mu} : = R_\mu \circ  \rho_{U/I , \mu} \circ q_{U/I} $.
\begin{align*}
\vcenter{
\xymatrix@C=10pt@R=4pt{
 A_\mathfrak{g}  / I(C)  \ar[r]^{ q_{U/I} \in Q}  \ar@/^20pt/[rr]^{ q_{\mu}^{pre} \in Q}
 \ar@/^40pt/[rrr]^{ q_{A/I \mu} \in Q}& 
 \mathcal{U}_\mathfrak{g}[\hbar ] / I(C(X)) \ar[r]^{\rho_{U/I, \mu}} &
  gl(V^\mu) \ar[r]^{R_\mu} &  gl(V^\mu)\\
   {}&{}&{}&{}\\
   {}\ar@{(-}[u] &{}\ar@{(-}[u] &{}\ar@{(-}[u] &{}\ar@{(-}[u] \\
\big[f(x)\big]=  \big[r_{f} + h_f \big]   \ar@{|->}[r]& 
[ q_U(r_f ) ]= \big[\sum_I a_I X_{(i_1, \cdots , i_m )} \big]  \ar@{|->}[r] &  
\sum_I a_I e^{(\mu)}_{(i_1 , \cdots , i_m)}
 \ar@{|->}[r] & \!\!\!\!\!
\sum_{|I| = m< n_\mu} \!\!\! a_I e^{(\mu)}_{(i_1 , \cdots , i_m)} 
}}
\end{align*}

It is not only that
the target space of the weak matrix regularization, $T_\mu$,
and this Lie-Poisson algebra $A_\mathfrak{g}  /I(C) $ are 
same structure as a Lie algebra.
In the sense of Proposition \ref{prop4_12},
Proposition \ref{prop_Kinji_2} or Proposition \ref{prop_Kinji_general_2}, 
$T_\mu$ is ``similar'' to  $A_\mathfrak{g}  /I(C) $ as an algebra 
in the limit as $\hbar(\mu)$ approaches zero.
The fact that the constructed weak matrix regularization possesses the approximate homomorphism property 
corresponds to the second condition in the general definition of matrix regularization (Definition \ref{matrixreg} in Appendix \ref{AppenB}).
In this sense, it possesses properties beyond the definition of a weak matrix regularization, 
exhibiting characteristics of a standard matrix regularization.\\

In Sections \ref{sect3} and \ref{sect4}, the discussions were restricted to irreducible representations, whereas in Section \ref{rev_sect6} 
it was shown that the construction can be fully generalized to reducible representations.
When confined to irreducible representations, the domain of each quantization admits kernels arising from geometry, so that the underlying space (the classical space) is effectively restricted to the corresponding coadjoint orbit.
However, by extending to reducible representations, it becomes possible to incorporate various orbits simultaneously.
In this case, the remaining issues are the problem of which sequence of representation spaces should be chosen and the problem of which region of the domain the weak matrix regularization fails to reflect.
These problems depend on the nature of the representations of each Lie algebra.
A detailed investigation of individual Lie algebras in this regard will be left for future work.\\

The eigenvalues of the Casimir operators are fixed as (\ref{fixed_casimir_relation})
in the limit where $\hbar(\mu)$ approaches $0$ and the dimension $\dim V^\mu$ 
approaches infinity.
Therefore, it would be natural to think of the variety determined 
by the Casimir polynomials as the classical space 
realized in the limit where $\hbar$ is zero.
However, the precision of this discussion of the classical(commutative) limit is a subject for future work.
\\

To know how different the weak matrix regularization constructed in this paper 
is from the matrix regularization written in Appendix \ref{AppenB} 
using the operator norm, we still need to introduce the operator norm 
and examine each of the conditions in Appendix \ref{AppenB}. 
This is another future work that has not yet been started.\\

In addition, there is no established method for relating 
``obtaining an effective theory on a classical manifold 
as the low-energy limit of the IKKT matrix model'' to ``the space of the corresponding commutative limit in matrix regularization''.
In fact, it was also seen in this paper that the manifold obtained 
in the commutative limit is not uniquely determined. 
Refining these discussions is also a future issue.


%
\section*{Acknowledgements}
\noindent 
A.S.\ was supported by JSPS KAKENHI Grant Number 21K03258.
The author also thanks the participants in the workshop
``Discrete Approaches to the Dynamics of Fields and Space-Time" 
for their useful comments.
We would like to thank A.~Tsuchiya and J.~Nishimura
for important information of the IKKT matrix model.
We are grateful to the anonymous reviewers of JMP for their constructive comments 
and insightful suggestions that helped improve the manuscript.
\\


\noindent
{\bf Data availability} \   Data sharing is not applicable to this article as no new data were 
created or analyzed in this study.

\section*{Declarations}

{\bf Conflicts of interest} \ On behalf of all authors, the corresponding author states that there is no conflict of
interest.

\appendix

\section{Definition of $\tilde{O}(z^{n} )$}\label{ap1}
Since we have not defined a norm on algebras
in this paper, the Landau symbol $O$ does not make sense. 
So, we define an order $\tilde{O}$ by $x\in \mathbb{R}$ 
using the absolute value of a complex number.

\begin{definition}
Let $\mathcal{V}$ be a vector space over $\mathbb{C}$. 
Let every $f_i$ be a complex valued continuous function such that
\begin{align*}
\lim_{x\to 0}  \left| \frac{f_i(xz)}{x^n} \right|  < \infty ,
\end{align*}
where $x\in \mathbb{R}$ and $z\in \mathbb{C}$. 
For $a_i\in \mathcal{V}$ which is independent of $z\in \mathbb{C}$, we denote the element described as $\sum_i f_i(z)a_i\in \mathcal{V}$ by 
$\tilde{O}(z^{n} )$.
If every  $f_i$ satisfies
\begin{align*}
\lim_{x\to 0}  \frac{f_i(xz)}{x^n}  =0 ,
\end{align*}
then we denote $\sum_i f_i(z)a_i\in \mathcal{V}$ by
$\tilde{O}(z^{n+\epsilon} )$.
\end{definition}
 Note that   $z$ itself is not necessarily continuous.
In this paper, the case where  $\mathcal{V}$ is an algebra often appears, but we are applying the above definition 
by considering it as a vector space.
For the purpose of this symbol, it is possible to replace 
$\tilde{O}(z^{n+1})$ with $\tilde{O}(z^n)$, 
but $\tilde{O}(z^{n})$ must not be replaced with $\tilde{O}(z^{n+1})$,
in the same way as for the usual ${O}(z)$.
\\

This definition can also be extended to the case of a quotient space as follows.  
\begin{definition}\label{defA2}
Let $\mathcal{V}$ be an algebra over a commutative ring $R$. 
For some $h \in \mathcal{V} , ~ [h] \in \mathcal{V} / \sim $,
if there exist $h' \in \mathcal{V} $ such that $[h] = [h']$ and 
$h' $ is $\tilde{O}(\hbar^{n})$, then we say
$[h] \in \mathcal{V} / \sim $ is  $\tilde{O}(\hbar^{n})$.
\end{definition}

\begin{example}\label{exA_3}
Let consider enveloping algebra $\mathcal{U}_{\mathfrak{g}} [\hbar ]$ with $X_i  X_j - X_j  X_i \sim \hbar f_{ij}^k X_k $ introduced in Subsection \ref{3_2}.
$[\hbar (X_i  X_j - X_j  X_i ) ] = \hbar^2 f_{ij}^k [X_k] $ is $\tilde{O}(\hbar^{2})$.
\end{example}

\begin{example}\label{exA_4}
Let us consider $\mathcal{U}_{\mathfrak{g}} (\hbar )/ I(C(X))$
in Subsection \ref{4_1}.
$[f(X)] \in \mathcal{U}_{\mathfrak{g}} (\hbar )/ I(C(X))$ is said to be $\tilde{O}(\hbar^{n})$
if there exists a $h(X)  \in \mathcal{U}_{\mathfrak{g}} (\hbar ) $ 
such that $[h(X)]=[f(X)]$ and $h(X)$ is $\tilde{O}(\hbar^{n})$ in the sense of Example \ref{exA_3}.
\end{example}

The following fact is proved in \cite{Sako:2022pid}.
\begin{proposition}\label{propA2}
Let $t_i : A \rightarrow M_i$ be a quantization defined by Definition \ref{Q}, and let $h_{ij}: M_i \rightarrow M_j$ be an $R$-algebra homomorphism.
Then 
$$
h_{ij} (\tilde{O}(\hbar^{1+\epsilon} (t_i )))= \tilde{O}(\hbar^{1+\epsilon} (t_i)) \in M_j.
$$
\end{proposition}

\section{Matrix regularization} \label{AppenB}
In this section, let us review the definition of standard matrix regularization for a symplectic manifold $(M,~\omega)$ in order to compare it with the definition given in this paper.
Matrix regularization \cite{matrix1} has evolved from the ideas of Berezin-Toeplitz quantization \cite{berezin1,berezin2}, Fuzzy space \cite{fuzzy1}, and so on. 
One mathematically sophisticated formulation is given in \cite{Rieffel:2021ykh}.
However, there is no unified common formulation.
Here, we introduce one of the definitions of matrix regularization as described in \cite{arnlind}, which is a widely known definition.
\begin{definition}\label{matrixreg}
Let $N_1,N_2,\ldots $ be a strictly increasing sequence of positive integers and $\hbar$ be a real-valued strictly positive decreasing function such that $\lim_{N\to \infty}N\hbar(N)$ converges. Let $T_k$ be a linear map from $C^\infty(M)$ to $N_k\times N_k$ Hermitian matrices for $k=1,2,\ldots$. If the following conditions are satisfied, then we call the pair $(T_k,~\hbar)$ a $C^1$-convergent matrix regularization of $(M,~\omega)$.
\begin{enumerate}
  \item $\displaystyle \lim_{k\to \infty}\|T_k(f)\|<\infty$,
  \item $\displaystyle \lim_{k\to \infty}\|T_k(fg)-T_{k}(f)T_{k}(g)\|=0$,
  \item $\displaystyle \lim_{k\to \infty}\|\frac{1}{i\hbar(N_k)}[T_k(f),T_k(g)]-T_k(\{f,g\})\|=0$,
  \item $\displaystyle \lim_{k\to \infty}2\pi\hbar(N_k){\rm Tr}T_k(f)=\int_M f\omega$,
\end{enumerate}
where $\|~\|$ is an operator norm, $\omega $ is a symplectic form on $M$ and $\{~,~\}$ is the Poisson bracket induced by $\omega$.
\end{definition}\par
As in this definition, one of the main differences between the many definitions of matrix regularization and the definition used 
in this paper is the introduction of the operator norm.
Definition \ref{matrixreg} also requires the recovery of homomorphism in the limit using 
the operator norm, whereas no such requirement is made in this paper.
The paper also imposes no restrictions on integrals or traces.
Overall, the definition of quantization in this paper is a less restrictive formulation.
Therefore, the term ``weak matrix regularization'' is used to distinguish it.
Originally, in the category-theoretic definition of the classical limit in \cite{Sako:2022pid}, 
we developed a formulation capable of describing various quantization procedures in a unified framework. 
For this reason, the introduction of a norm is deliberately avoided, and the space is kept as structure-free as possible.
This explains why fewer conditions are imposed compared to many definitions of matrix regularization.


\section{A brief summary of Gr\"obner basis} \label{appendix_grobner}

The definitions of words and phrases related to the Gr\"obner basis and some of its properties are summarized in this appendix.
(See for example \cite{Cox_Little_Oshea}. )

Let $K$ be a field. In this paper we consider $K = {\mathbb C}$.
Let $K[x_1, \dots, x_n]$ be a polynomial ring with some fixed order.
For example, the graded lexicographic ordering for monomials 
$X^{\alpha} = x_1^{\alpha_1} x_2^{\alpha_2} \cdots x_n^{\alpha_n} $
and 
$X^{\beta} = x_1^{\beta_1} x_2^{\beta_2} \cdots x_n^{\beta_n} $
($\alpha = (\alpha_1, \dots, \alpha_n),  \beta = (\beta_1, \dots, \beta_n)$)
 is given by
\[
\alpha < \beta \xLeftrightarrow[]{def} 
\begin{cases}
deg( X^{\alpha} ):= \alpha_1 +\cdots + \alpha_n < deg ( X^{\beta} ):=  \beta_1 + \cdots + \beta_n  \\
\alpha_1 = \beta_1,
\dots ,
\alpha_{i-1} = \beta_{i-1} ,  \alpha_i < \beta_i ~
 \mbox{ when}~
\alpha_1 +\cdots + \alpha_n = \beta_1 + \cdots + \beta_n   
\end{cases} .
\]
We employ the graded lexicographic ordering as the fixed order in this paper.\\

Next, we introduce some terms to define the Gr\"obner basis.
\begin{definition}
Let $
f = \sum_{\alpha} a_{\alpha} x^{\alpha} \quad \left( \alpha = (\alpha_1, \alpha_2, \dots, \alpha_n) , ~ x^{\alpha} = x_1^{\alpha_1} x_2^{\alpha_2} \cdots x_n^{\alpha_n}\right)
$ is an element of  $K[x_1, \dots, x_n]$ with some fixed order.
\begin{enumerate}
\item We say that $ \max \{\alpha \mid a_{\alpha} \neq 0 \}$  is multi-degree of $f$, 
and we denote it by $ \text{multdeg}(f) $.

\item Leading Monomial :
${LM}(f) := x^{\text{multdeg}(f)}$ is called the leading monomial
of $ f $ .

\item Leading Coefficient :
${LC}(f) := a_{\text{multdeg}(f)} $ is called the leading coefficient of $ f $.

\item  Leading Term :
${LT}(f) := \text{LC}(f) \cdot {LM}(f) $ is called leading term  of $ f $.

\end{enumerate}

\end{definition}

In addition, we introduce following symbols
for some subset $S \subset K[x_1, \dots, x_n]$.
$
{LM}(S) := \{ {LM}(f) \mid f \in S \}
$.
We denote the monomial ideal generated by $LM(S)$ by
$
\langle \text{LM}(S) \rangle 
$.
We also use
$
{LT}(S) := \{ {LT}(f) \mid f \in S \},
$
and the ideal generated by $LT(S)$ is denoted by
$\langle \text{LT}(S) \rangle $.
Therefore, we find
\[
\langle \text{LM}(S) \rangle = \langle \text{LT}(S) \rangle .
\]

\begin{definition}
Let $I$ be an ideal of $K[x_1, \dots, x_n]$.
We say that $G = \{f_1, \dots, f_s\} \subset I$ 
is a Gr\"obner basis if
\[
\langle \text{LM}(I) \rangle = \langle \text{LM}(G) \rangle = \langle \text{LM}(f_1), \dots, \text{LM}(f_s) \rangle .
\]
\end{definition}

In the following, we list some important properties about the Gr\"obner basis
\cite{Dumniit_Foote_Abstract Algebra,Cox_Little_Oshea}.

\begin{proposition}
For any monomial ordering and any ideal $I$ that is not $\{ 0 \}$, there exists a Gr\"obner basis of $K[x_1, \ldots, x_n]$ that generates $I$.
\end{proposition}

\begin{proposition}
Let $G$ be a Gr\"obner basis of $K[x_1, \ldots, x_n]$ that generates $I$.
For any $ f \in K[x_1, \ldots, x_n] $, 
\begin{enumerate}
\item There exists a polynomial  $g \in I$, and a polynomial $r \notin I$ 
satisfying $f = g + r$,
such that any monomial in $r$ is not divisible by any element of $LT (G)$.
We call $r$ the remainder.
\item The above $g$ and $r$  are uniquely determined by $f$ and $G$.
\end{enumerate}
\end{proposition}

It is necessary to be careful, as the result may change depending on the choice of the Gr\"obner basis $G$.
However,
 it is possible to introduce a good Gr\"obner basis that fixes the arbitrariness of 
the choice of $G$.
\begin{definition}
If a Gr\"obner basis $G = \{ f_1, \dots, f_s \} \subset I$
satisfies the following two conditions 
i) ${LC}(f_i) = 1$ for all $i$ , ~ ii)  no term in
$f_i \in  G$ is divisible by $LM(f_j ) (i\neq j)$,
then we say that $G$ is a reduced  Gr\"obner basis.
\end{definition}

\begin{theorem}
Let $I \neq \{0\}$ be a polynomial ideal. 
Then, for a given monomial ordering,
there is a reduced Gr\"obner basis for $I$, and it is unique.
\end{theorem}

The fact that Theorem \ref{reducedGrobner} holds for a reduced Gr\"obner basis 
defined above is used in this paper.



\end{document}